\documentclass[lettersize,journal,a4paper]{IEEEtran}

\usepackage{amsmath}
\usepackage{amsthm}

\theoremstyle{plain}

\newtheorem{remark}{Remark}[section]

\newtheorem{proposition}{Proposition}[section]

\newtheorem{approximation}{Approximation}

\newtheorem{corollary}{Corollary}[section]

\usepackage{color}
\usepackage{graphicx} 
\usepackage{bm}
\usepackage{multicol}
\usepackage{setspace}
\usepackage{subfigure}
\usepackage{fancyhdr} % 添加页眉页脚
\usepackage{indentfirst}
\usepackage{enumerate}
\usepackage{amssymb}
\usepackage{stfloats}
\usepackage{bm}
\usepackage{threeparttable}
\usepackage{CJK}
\usepackage[justification=centering]{caption}
\usepackage{makecell}
\usepackage{caption}
\usepackage{cases}
\usepackage{algorithm}
\usepackage{algpseudocode}
\usepackage{graphics}
\usepackage{epsfig}
\usepackage{cite}
\usepackage{url}
\begin{document}
	\title{SpaceDiffusion: Over-the-Orbit Diffusion for Space Generate-and-Forward Communications}
	\author{Jianhao~Huang, Zhanwei~Wang, Khaled~B.~Letaief,~\IEEEmembership{Fellow,~IEEE},
		and~Kaibin~Huang,~\IEEEmembership{Fellow,~IEEE}%
		\thanks {This work has been submitted to the IEEE for possible publication.
			Copyright may be transferred without notice, after which this version may
			no longer be accessible.}
		\thanks{J. Huang, Z. Wang, and K. Huang are with the Department of Electrical
and Computer Engineering, The University of Hong Kong (HKU), Hong
Kong SAR, China (Email: Jianhaoh@hku.hk, \{zhanweiw, huangkb\}@eee.hku.hk).
Corresponding author: K. Huang.}%
		\thanks{Khaled B. Letaief is with the Department of Electronic and Computer
Engineering, The Hong Kong University of Science and Technology (HKUST),
Hong Kong SAR, China (Email: eekhaled@ust.hk).}}

	\maketitle

\begin{abstract}
Satellite communications are an essential component of sixth-generation (6G) mobile networks, which provide ubiquitous connectivity for global services. However, the satellite uplink remains a critical bottleneck for ground devices: their limited transmit power and antenna apertures result in low data rates and high packet errors.  To overcome this bottleneck, this paper advocates a novel relaying paradigm termed generate-and-forward (GF) communications, where satellites exploit on-orbit generative artificial intelligence (AI) to robustly reconstruct corrupted data prior to forwarding. Specifically, we propose \emph{SpaceDiffusion}, an over-the-orbit diffusion framework for satellite-assisted image transmission. The core of this framework is a channel-distortion-aware diffusion theory developed using the following approach. By formulating the recovery of compressed and lost image tokens as an inverse problem, this theory incorporates a channel-distortion correction term directly into the conventional denoising diffusion implicit model (DDIM) update. As a result, this design enables a single pretrained diffusion model to adapt dynamically to varying packet-loss patterns and compression distortions without retraining. Furthermore, we analytically characterize the progressive token-reconstruction error and derive a diffusion-step activation threshold that predicts when SpaceDiffusion is expected to outperform conventional decode-and-forward (DF) relaying. Building on these theoretical insights, we further develop an energy-aware early-exit policy to efficiently deploy SpaceDiffusion in orbit.  Experimental results demonstrate that SpaceDiffusion achieves lower end-to-end latency  compared to DF scheme with  retransmission protocol and saves approximately $15$ dB of uplink transmit power at a target perceptual quality. 
\end{abstract}

\begin{IEEEkeywords}
Satellite communications, generative artificial intelligence, generate-and-forward, diffusion models.
\end{IEEEkeywords}

\section{Introduction}

Satellite communications are envisioned as an essential component of sixth-generation (6G) mobile networks, which enable ubiquitous coverage for global services such as artificial intelligence (AI), immersive Internet of Things (IoT), and multimedia access \cite{10560514,11143883,10835069,kodheli2021satcomnewspace,11569083}. In conventional architectures, satellites are often regarded as spaceborne extensions of terrestrial relaying nodes, and their primary role is to receive and forward information between ground terminals \cite{kodheli2021satcomnewspace}. However, this relay-centric design faces inherent limitations in satellite uplinks due to long propagation distances, Doppler shifts, and limited transmit power of ground devices. To address the issue, fixed ground terminals have been provisioned with high-gain antennas, advanced beamformers, and higher-frequency spectrum resources, such as Ku- and Ka-band \cite{3gpp_tr38811,kodheli2021satcomnewspace}. However, these solutions are impractical for IoT and handheld devices as they are constrained by battery capacity, hardware complexity, and antenna size \cite{sorensen2021nbiotleo}. As a result, the low uplink transmission rates (typically ranging from a few to tens of Kbps for low-power devices) remain a critical bottleneck for bandwidth-demanding multimedia services \cite{sorensen2021nbiotleo,Talgat2024UplinkLEOIoT}. Recent advances in space AI and on-orbit computing have been transforming satellites from passive forwarding nodes into intelligent information-processing centers \cite{wang2026spacemoE,feilden2024train,zhejianglab2025threebody}. For example, several on-orbit AI computing platforms, such as Starcloud-1 \cite{feilden2024train} and the Three-Body Computing Constellation \cite{zhejianglab2025threebody}, have already been launched. This paradigm shift opens new opportunities to reduce the need of high-throughput uplink transmission by enabling on-orbit generative  data recovery.

\begin{figure}[t]
	\normalsize
	\setlength{\textfloatsep}{6pt plus 2pt minus 2pt}
	\setlength{\abovecaptionskip}{4pt}
	\setlength{\belowcaptionskip}{-0.1cm}
	\centering
	\includegraphics[width=0.97\linewidth]{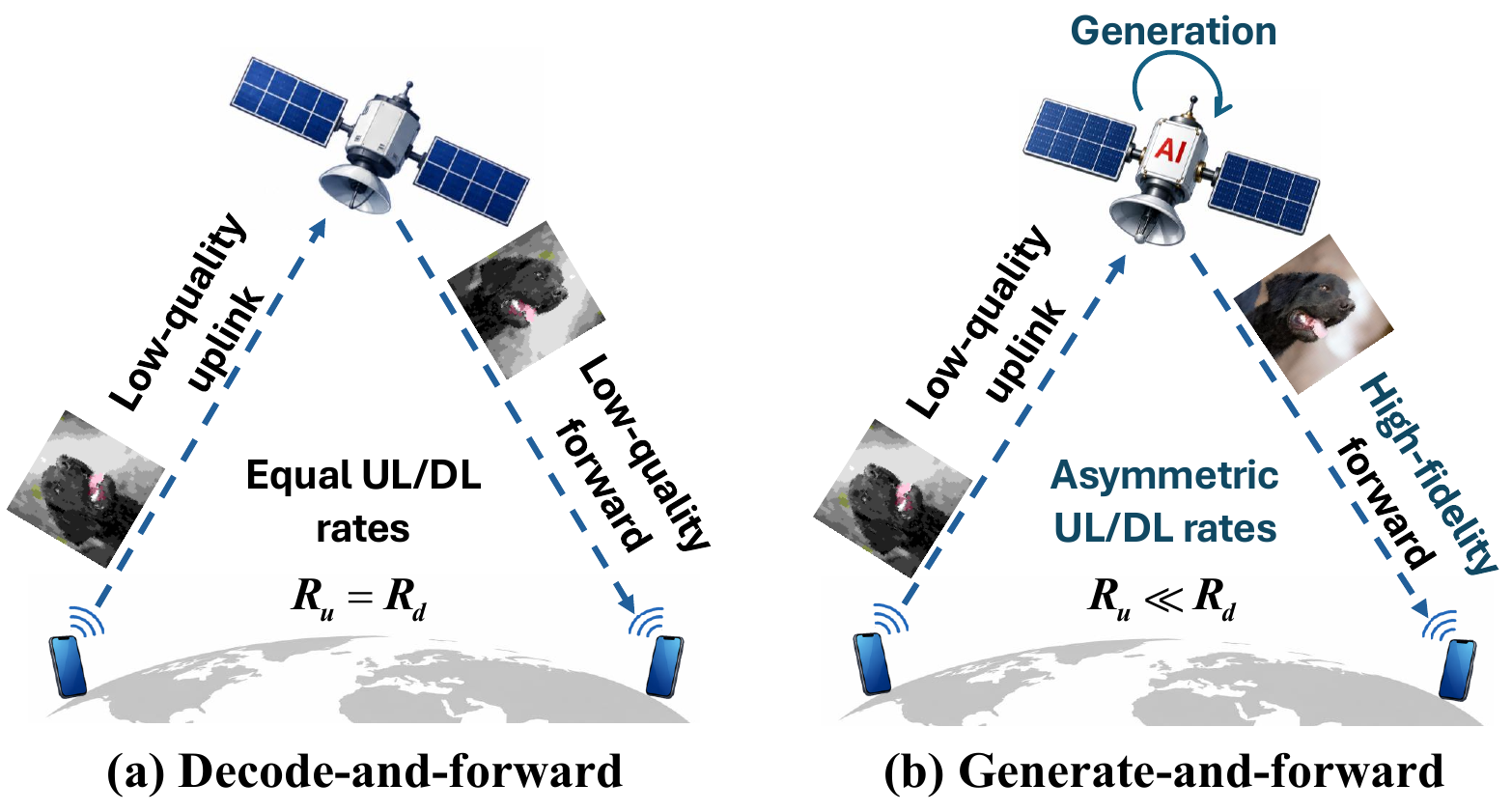}
	%\caption{fig2}
	\captionsetup{justification=justified}
		\caption{Comparison of the satellite generate-and-forward scheme with the conventional decode-and-forward scheme. For satellite communications with limited-power devices, the downlink capacity (on the order of Mbps) is much larger than the uplink capacity (on the order of Kbps) \cite{11589398}.}
	\label{architecture}
\end{figure}
Under this uplink bottleneck, conventional satellite systems typically rely on either \emph{transparent} or \emph{regenerative} architectures to support reliable transmission \cite{3gpp_tr38811,3GPPTR38821}. On one hand, the transparent architecture adopts the low-complexity amplify-and-forward (AF) scheme, where the satellite directly forwards the received signals after power amplification and frequency conversion \cite{3gpp_tr38811,kodheli2021satcomnewspace}. On the other hand, the regenerative architecture is based on the decode-and-forward (DF) scheme, where the satellite performs onboard channel decoding and re-encoding, thereby mitigating error propagation \cite{3gpp_tr38811,3GPPTR38821}. Nevertheless, both architectures still depend on accurate bit-level delivery over the uplink. When packets are corrupted or fail to decode, acknowledgement (ACK)/negative acknowledgement (NACK) feedback and retransmission are required to recover the lost information \cite{6888474,3GPPTR38821}. Due to the long round-trip time (RTT) of satellite links, such feedback-based retransmission mechanisms can significantly increase end-to-end (E2E) latency. Existing studies attempt to mitigate RTT-induced stalling through enhanced retransmission scheduling, advanced forward error correction, and increasing the number of parallel hybrid automatic repeat request (HARQ) processes \cite{3GPPTR38821}. However, these methods still operate within the feedback-based retransmission paradigm and thus cannot fundamentally overcome the mentioned limitations. This necessitates a paradigm shift from accurate bit delivery to more efficient AI-empowered data recovery. 

Recently, generative AI (GenAI) has achieved remarkable success in data recovery from low-dimensional prompts \cite{Song2021DDIM,rombach2022high,chung2023diffusion}. Motivated by this capability, existing studies have introduced GenAI into terrestrial communication networks to improve transmission efficiency and robustness against channel impairments \cite{wang2025resicomp,qiao2025token,huang2026generative,jiang2025semantic}. The communication efficiency can be improved by converting source data into low-dimensional feature or token domain for transmission, followed by generative recovery at the receiver.  For example, transformer-based foundation models have been used to reconstruct latent tokens lost to packet erasures \cite{wang2025resicomp,qiao2025token}, while diffusion models have been used to iteratively impute missing features for error-resilient transmission \cite{huang2026generative}. These generative techniques can be easily extended to satellite communications by treating the satellite as a conventional relay node. A representative example involves deploying a foundation model at the ground gateway to recover corrupted features forwarded by the satellite \cite{jiang2025semantic,95ab4f9188f548208bf5b1e244e39dfe}. However, such generative strategy falls short on exploiting satellite computation resources to overcome the uplink bottleneck.  Moreover, when key features are lost, gateway-side recovery requires an additional uplink and downlink retransmission cycle, which substantially increases the latency.

With the advancement in space AI computing, satellites can perform generative data processing on-orbit before forwarding \cite{11589398,shi2025satellite}. This alleviates the burden on resource-constrained uplink transmission.  We refer to this new relaying paradigm as the \emph{generate-and-forward} (GF) scheme, as illustrated in Fig. \ref{architecture}. The GF scheme leverages GenAI to recover missing features/tokens or lost information caused by packet error, rate control, or source compression under low-rate uplink transmissions. Consequently, the generated data allows the flexibility of having more detailed representations than the transmitted one. The resultant increased forwarding payload can be supported by leveraging the asymmetric uplink-downlink capacities and relaxing the conventional uplink-downlink rate-matching constraint. However, the realization of GF scheme in the satellite system  faces two new challenges.  The first is the design of efficient generative algorithms tailored to satellite communications.  Existing generative models are predominantly data-driven and may not generalize well to dynamic satellite channel environments \cite{rombach2022high,Song2021DDIM}. The second is the deployment of potentially power-hungry  GenAI algorithms on energy-constrained satellites, where the available energy harvested from solar power varies with orbital dynamics and illumination conditions.

In this paper, we address these challenges by proposing \emph{SpaceDiffusion}, a novel over-the-orbit diffusion framework for space-based GF image transmission. The techniques in the framework alleviate the uplink bottleneck by enabling the satellite to recover highly compressed and partially erased image tokens on board.  The main contributions and findings of this study are summarized as follows.
\begin{itemize}
	
	\item \textbf{Channel-distortion-aware diffusion theory:} We formulate the problem of satellite token recovery under packet loss and lossy source compression as an inverse problem.  Based on this model, we derive a posterior-based reverse-time stochastic differential equation (SDE) that progressively recovers clean tokens from Gaussian noise. 
	The core of the derived SDE is the posterior score function that decomposes into a data-distribution term and a channel-distortion term determined by the packet-loss mask and source-compression noise covariance.  By exploiting this additive decomposition, we develop a channel-distortion-aware denoising diffusion implicit model (DDIM) algorithm that inserts the channel-distortion term as a plug-in correction to the conventional DDIM update. This enables a single pre-trained diffusion model to adapt to varying satellite channel conditions without retraining.
	
	\item \textbf{Diffusion performance analysis:} We analyze the progressive performance of the proposed channel-distortion-aware DDIM as the reverse-step index $t$ decreases from $T$ to $0$, where $T$ is the maximum reverse-step index. Leveraging the pseudo-inverse Gaussian approximation and the forward diffusion relation, we derive an approximate upper bound on the token-domain mean-squared error (MSE) at each step.  Based on this bound and the progressive denoising property of reverse diffusion \cite{Song2021DDIM,Song2021ScoreBasedSDE}, we further derive a diffusion-step threshold that predicts the minimum number of reverse steps required for SpaceDiffusion to outperform conventional DF relaying.
	
	\item \textbf{Energy-aware diffusion-step optimization:} Based on the derived threshold, we formulate an optimization problem that maximizes the number of executable reverse diffusion steps under energy and latency constraints. By solving this problem, we derive a closed-form early-exit policy that activates on-board generation only when the available energy and time are sufficient to reach the diffusion step threshold; otherwise, the satellite falls back to DF relaying to avoid ineffective diffusion computation.

	\item \textbf{Experimental results:} Extensive experiments on the real-world AFHQ and Kodak image datasets validate SpaceDiffusion under diverse channel and orbital conditions. Compared with a DF baseline employing  retransmission, SpaceDiffusion reduces E2E latency by up to \(40\%\) in LEO. In addition, SpaceDiffusion requires about \(15\) dB less uplink transmit power than the non-generative compression baseline to achieve the same perceptual quality. This power efficiency makes it well suited to satellite applications, such as IoT connectivity over non-terrestrial networks (NTNs).  Experiments across different orbital configurations further show that orbital illumination strongly affects the number of executable diffusion steps, while the proposed energy-aware policy adapts onboard generation to the available resources and consistently outperforms the no-generation baseline.

\end{itemize}

The remainder of this paper is organized as follows. Section II gives an overview of the satellite GF communications. Section III presents the system blocks and problem formulations. Section IV develops the proposed posterior-based diffusion theory and the algorithm. Section V studies optimal diffusion steps under satellite energy limitations. Section VI provides experimental results, and Section VII concludes this paper.

\section{Overview of Satellite Generate-and-Forward Communications}

Consider a satellite GF system, where a  ground device uploads data to a satellite, and the satellite performs on-board generative recovery before forwarding the information to a ground device/station. 
The satellite is equipped with a GenAI-enabled module to reconstruct high-fidelity data from corrupted inputs due to poor satellite uplink transmissions.

Specifically, the overall operation workflow consists of four sequential stages.
\begin{enumerate}
	\item \textbf{Tokenization:} Let
	$\mathbf{X}$
		denote the source data to be transmitted. At the transmitter, the semantic encoder $F_{\bm{\phi}}(\cdot)$ transforms $\mathbf{X}$ into a token representation
		$
				\mathbf{Y} \in\mathbb{R}^{K\times D},
		$
		where $K$ denotes the number of tokens and $D$ is the token dimension. The token representation $\mathbf{Y}$ is then quantized and source-coded into bit streams.

	\item \textbf{Uplink transmission:}
		The compressed token bitstream is  interleaved and packetized into a sequence of fixed-length packets for transmission.  Due to the poor satellite channel conditions and limited bandwidth, uplink transmission suffers from high compression  and  packet losses. 
	
	\item \textbf{Over-the-orbit token generation:} After uplink reception and decoding, the satellite obtains a distorted token representation
$
		\tilde{\mathbf{Y}}\in\mathbb{R}^{K\times D},
$
		in which some tokens are correctly received while others are  missing due to packet loss. We propose an over-the-orbit generation module $G_{\bm{\Theta}}(\cdot)$ to recover the  corrupted tokens directly in the latent space to obtain $\bar{\mathbf{Y}}$.
		
		\item \textbf{Downlink forwarding:} After the generation process, the satellite re-encodes and forwards the recovered information to the destination ground user through the downlink.  The destination user reconstructs the image through a semantic decoder $D_{\bm{\psi}}(\cdot)$.
\end{enumerate}

In the following section, we elaborate on the above four stages in detail and then discuss the main challenges of the  GF communication system.

\begin{figure}[t]
	\normalsize
	\setlength{\abovecaptionskip}{4pt}
	\setlength{\belowcaptionskip}{-0.1cm}
	\centering
	\includegraphics[width=0.98\linewidth]{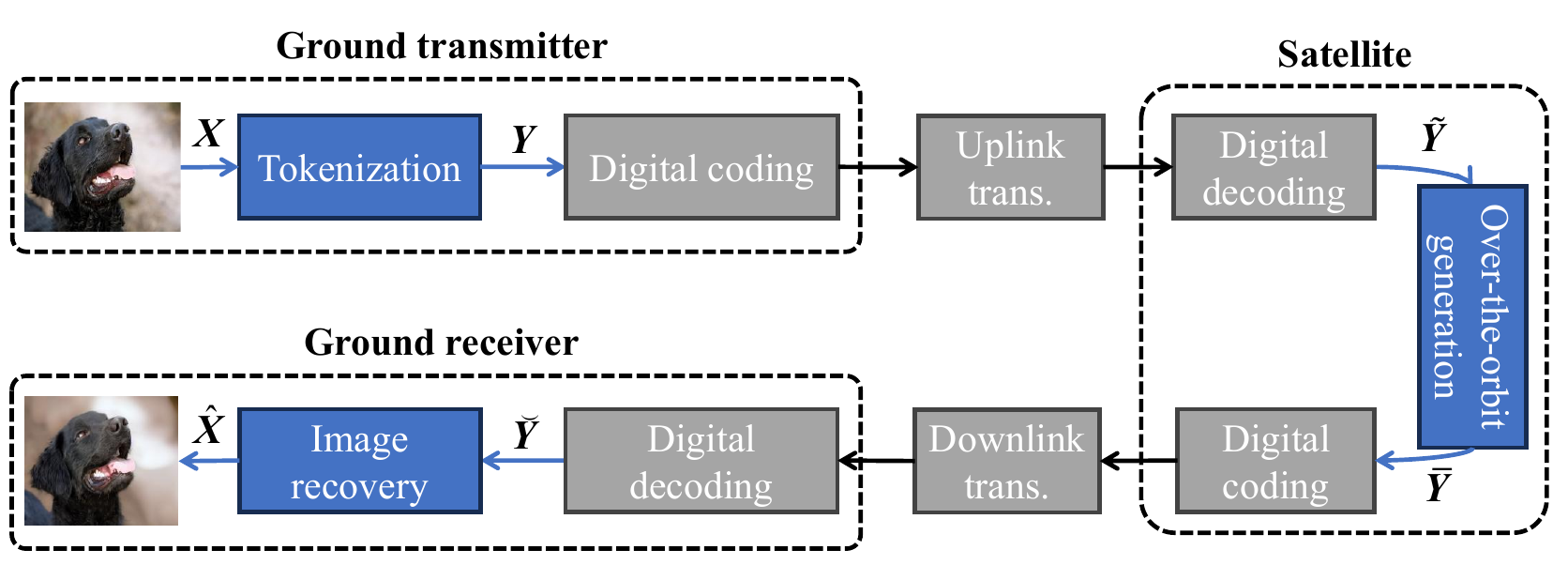}
	%\caption{fig2}
	\captionsetup{justification=justified}
	\caption{ The framework of the Satellite GF system}
	\label{system}
\end{figure}

\section{System Model and Problem Formulation}
In this section, we present the system model of the GF  system and introduce its main  problems. Here, we take image transmission as an example, but the proposed framework can be extended to other data types, such as video and audio, as shown in Fig. \ref{system}. 
\subsection{Tokenization}

\subsubsection{Token Representation}

For a given source image $\mathbf{X}\in\mathbb{R}^{H\times W\times 3}$, the transmitter employs a neural network-based encoder $F_{\bm{\phi}}(\cdot)$ to directly map the image into a token representation, i.e.,
\begin{equation}
	\mathbf{Y}=F_{\bm{\phi}}(\mathbf{X})\in\mathbb{R}^{K\times D}.
\end{equation}
The token representation can be written as
\begin{equation}
	\mathbf{Y}=
	\left[
	\mathbf{y}_1,\mathbf{y}_2,\ldots,\mathbf{y}_{K}
	\right]^{\mathsf T}.
\end{equation}
Here, $\mathbf{y}_n\in\mathbb{R}^{D}$ denotes the $n$-th token, which represents a local latent patch of the source image. Tokenization converts the image into a sequence of  local features, thereby reducing pixel-level redundancy and enabling efficient compression and packetization \cite{wang2025resicomp}. This local token representation is also well suited for the diffusion model in Section IV.

\subsubsection{Digital Source Coding}
Due to the limited uplink bandwidth, the token representation $\mathbf{Y}$ needs to be compressed in a lossy manner before transmission. The design of an efficient lossy source coder generally requires knowledge of the token distribution. However, the exact distribution of $\mathbf{Y}$ is difficult to characterize due to the intractable image distribution and the nonlinear semantic encoder. To facilitate tractable source-coding design, we assume that the tokens are independent and identically distributed (i.i.d.) and that each token $\mathbf{y}_n$ approximately follows a multivariate Gaussian distribution\footnote{This Gaussian approximation is introduced only for the source coding design in this subsection. It is not assumed in the later diffusion-based recovery process, where the accuracy of token prior plays a critical role in the reconstruction quality.}, i.e.,
\begin{equation}
	\mathbf{y}_n \sim \mathcal{N}(\boldsymbol{\mu},\boldsymbol{\Sigma}),
\end{equation}
where $\boldsymbol{\mu}\in\mathbb{R}^{D}$ and $\boldsymbol{\Sigma}\in\mathbb{R}^{D\times D}$ denote the mean vector and covariance matrix estimated from training samples. The covariance matrix $\boldsymbol{\Sigma}$ captures the statistical correlations among token components. 

% Specifically, the covariance matrix admits the eigenvalue decomposition
% where $\mathbf{U}$ is an orthogonal transform matrix and
% contains the eigenvalues. Each token is first transformed into the decorrelated domain as
% where $\mathbf{z}_n=[z_{n,1},\ldots,z_{n,D}]^{\mathsf T}$ contains approximately decorrelated Gaussian components.
Here, we adopt linear transform coding \cite{Huang1963BlockQuantization} to encode the correlated Gaussian sources, as it naturally exploits their second-order correlation structure while retaining low implementation complexity. Moreover, it allows the source coding rate to be conveniently adapted to the uplink channel conditions by adjusting the quantization resolution. Specifically,
the eigendecomposition of the covariance matrix is
$
	\boldsymbol{\Sigma}=\mathbf{U}\boldsymbol{\Lambda}\mathbf{U}^{\mathsf T},
$
where $\mathbf{U}$ is orthogonal and
$
	\boldsymbol{\Lambda}=\mathrm{diag}(\lambda_1,\lambda_2,\ldots,\lambda_D)
$
contains the eigenvalues. Each token is decorrelated as
$
	\mathbf{z}_n=\mathbf{U}^{\mathsf T}(\mathbf{y}_n-\boldsymbol{\mu}),
$
where $\mathbf{z}_n=[z_{n,1},\ldots,z_{n,D}]^{\mathsf T}$.
Each coefficient is uniformly quantized \cite{gray1998quantization} as
\begin{equation}
\label{quanti:1}	\hat{z}_{n,k}
	=
	\Delta_k
	\left\lfloor
	\frac{z_{n,k}}{\Delta_k}
	+\frac{1}{2}
	\right\rfloor,
\end{equation}
where $\lfloor \cdot \rfloor$ is the floor operator and
$
	\Delta_k=\frac{\sqrt{\lambda_k}}{s},
$
with $s>0$ controlling the quantization resolution. Huffman coding \cite{Huffman1952MinimumRedundancy} then encodes each value using approximately
$
	B_{n,k}\approx -\log_2 P_k(\hat{z}_{n,k}),
$
bits, where $P_k$ is the probability mass function. 

% After entropy decoding, the reconstructed token is obtained through the inverse transform
% According to \eqref{quanti:1} and \eqref{trans:1}, the reconstructed token can be expressed as
% where $\mathbf{q}_n=\mathbf{U}\mathbf{e}_n$ and $\mathbf{e}_n=\hat{\mathbf{z}}-\mathbf{z}_n$ is the quantization noise. When the token dimension $D$ is sufficiently large, by the Central Limit Theorem \cite{rohatgi2015introduction}, $\mathbf{q}_n$ approximately follows a Gaussian distribution, i.e., $\mathbf{q}_n \sim \mathcal{N}(\mathbf{0},\boldsymbol{\Sigma}_q)$ with covariance
After entropy decoding, the inverse transform gives
$
	\hat{\mathbf{y}}_n=\mathbf{U}\hat{\mathbf{z}}_n+\boldsymbol{\mu}.
$
According to \eqref{quanti:1}, we have  
\begin{equation}
	\hat{\mathbf{y}}_n=\mathbf{y}_n+\mathbf{q}_n,
	\label{eq:additive_token_noise}
\end{equation}
where $\mathbf{q}_n=\mathbf{U}\mathbf{e}_n$ and
$
	\mathbf{e}_n=\hat{\mathbf{z}}_n-\mathbf{z}_n
$ is the quantization error. For sufficiently large $D$, the Central Limit Theorem \cite{rohatgi2015introduction} gives $\mathbf{q}_n\sim\mathcal{N}(\mathbf{0},\boldsymbol{\Sigma}_q)$, where
\begin{equation}
	\boldsymbol{\Sigma}_q
	=
	\mathbf{U}
	\mathrm{diag}\!\left(
	\frac{\Delta_1^2}{12},\ldots,\frac{\Delta_D^2}{12}
	\right)
	\mathbf{U}^{\mathsf T}
	=
	\frac{1}{12s^2}\boldsymbol{\Sigma}.
\end{equation}

\subsection{Uplink Transmissions}

After tokenization and source coding, the compressed bitstream is transmitted from the  ground device to the relay satellite. 

\subsubsection{Interleaving \& Packetization}

Denote the source-coded bitstream of an image as $\mathbf{b}=\{\mathbf{b}_1,\mathbf{b}_2,\ldots,\mathbf{b}_{K}\}$, where $\mathbf{b}_n$ is the compressed bit sequence of the $n$-th token. To mitigate the impact of burst errors during transmission, an interleaving operation is applied to the sequence. Defined by a permutation $\pi(\cdot)$, the interleaved bitstream is $\mathbf{b}^{\pi}=\{\mathbf{b}_{\pi(1)},\mathbf{b}_{\pi(2)},\ldots,\mathbf{b}_{\pi(K)}\}$.

After interleaving, the tokens are assembled into fixed-length packets of length $L_{\rm p}$. Each packet comprises a header of length $L_h$ and a payload of length $L_{\mathrm{pay}} = L_{\rm p}-L_h$. To prevent error propagation \cite{10845799}, each token's bitstream is treated as an indivisible  unit, meaning a token cannot be fragmented across multiple packets. If the allocated token bits are less than the payload capacity, the packet is simply padded with zeros.
Let $\mathcal{I}_m$ denote the ordered set of token indices allocated to the $m$-th packet. The $m$-th packet is constructed as
\begin{equation}
	\mathbf{p}_m=
	\big[
	\mathbf{h}_m,\,
	\mathbf{b}_{\pi(i_1)},\,
	\ldots,\,
	\mathbf{b}_{\pi(i_k)},\,
	\mathbf{0}
	\big],
\end{equation}
where $\mathbf{h}_m$ is the header, $i_k \in \mathcal{I}_m$, and $\mathbf{0}$ denotes zero-padding bits. 
 The header carries the control information required for token reconstruction, including the ordered token indices, and the padding length. Accordingly, the final packet sequence is denoted by $\mathcal{P}=\{\mathbf{p}_1,\mathbf{p}_2,\ldots,\mathbf{p}_M\}$, where $M$ denotes the total number of packets.

\subsubsection{Channel Coding}

% Each source packet $\mathbf{p}_m$ is further protected by a digital channel coding. Let $\mathcal{C}(\cdot)$ denote the channel coding function. The coded symbols corresponding to $\mathbf{p}_m$ is expressed as
% where $L_c$ denotes the block length. The corresponding channel coding rate in bits per channel use is given by
% With uplink bandwidth $B$, the uplink transmission rate in bits per second is
Each packet is channel-coded as
\begin{equation}
	\mathbf{c}_m=\mathcal{C}(\mathbf{p}_m) \in \mathbb{C}^{L_c},
\end{equation}
where $\mathcal{C}(\cdot)$ is the encoder and $L_c$ is the block length. The code rate is
$
	r_{\rm u}=\frac{L_{\rm p}}{L_c}.
$
Thus, the uplink bit rate over bandwidth $B$ is
$
	R_{\rm u}=B r_{\rm u}.
$

\subsubsection{Uplink Model}
Consider a block fading channel model \cite{Talgat2024UplinkLEOIoT}, where the received signal at the satellite for packet $m$ is given by
% where $p_u>0$ is the uplink transmission power, $h_m\in\mathbb{C}$ denotes the uplink channel coefficient, and $\mathbf{w}_m\sim\mathcal{CN}(\mathbf{0},\sigma^2\mathbf{I})$ is the additive white Gaussian noise vector.
\begin{equation}
	\mathbf{r}_m=h_m \sqrt{p_u}\,\mathbf{c}_m+\mathbf{w}_m,
\end{equation}
where $p_u$ is the transmit power, $h_m$ is the channel coefficient, and $\mathbf{w}_m\sim\mathcal{CN}(\mathbf{0},\sigma^2\mathbf{I})$ is additive noise.

% To capture the propagation characteristics of satellite communications, the uplink channel coefficient is modeled as
% where $\ell_m$ denotes the large-scale channel gain and $g_m$ is the small-scale Rician fading coefficient. The large-scale gain mainly depends on the free-space path loss, antenna gains, and atmospheric attenuation, and is expressed as
The channel is decomposed as
$
	h_m=\sqrt{\ell_m}\,g_m,
$
where $\ell_m$ is the large-scale gain and $g_m$ is the Rician fading coefficient. The former is
\begin{equation}
	\ell_m=\frac{G_tG_r}{L_{\mathrm{fs}}(d_m,f_c)\xi_m},
\end{equation}
% where $G_t$ and $G_r$ are the transmit and receive antenna gains, respectively, $d_m$ is the propagation distance, $f_c$ is the carrier frequency, $L_{\mathrm{fs}}(d_m,f_c)$ denotes the free-space path loss, and $\xi_m$ accounts for additional attenuation such as atmospheric absorption, rain fading, and shadowing. The received signal-to-noise ratio (SNR) for packet $m$ is defined as
where $G_t$ and $G_r$ are the antenna gains, $d_m$ is the propagation distance, $f_c$ is the carrier frequency, $L_{\mathrm{fs}}$ is the free-space path loss, and $\xi_m$ captures additional attenuation. The received SNR is
$
	\gamma_m = \frac{p_u |h_m|^2}{\sigma^2}.
$
% Under finite-blocklength channel coding, the decoding error probability of packet $m$ can be approximated by \cite{polyanskiy2010channel}
% where $C(\gamma_m)=\log_2(1+\gamma_m)$ is the channel capacity in bits per channel use, $V(\gamma_m)=\left(1-(1+\gamma_m)^{-2}\right)(\log_2 e)^2$ is the channel dispersion, and $Q(\cdot)$ denotes the Gaussian $Q$-function \cite{polyanskiy2010channel}.
With block length $L_c$, the decoding error probability is approximated by \cite{polyanskiy2010channel}
\begin{align} \label{finite_eq}
	\rho_m \approx Q\!\left(
			\frac{\sqrt{L_c}\left(C(\gamma_m)-r_{\rm u}\right)}
	{\sqrt{V(\gamma_m)}}
		\right),
\end{align}
where $C(\gamma_m)=\log_2(1+\gamma_m)$, $V(\gamma_m)=\left(1-(1+\gamma_m)^{-2}\right)(\log_2 e)^2$, and $Q(\cdot)$ denote the channel capacity, channel dispersion, and Gaussian $Q$-function, respectively.

%
%\subsubsection{Latency Constraint}
%
%Due to the short satellite visibility window, the uplink transmission must be completed within the available contact duration $T$. Let $L_o$ denote the average packet overhead, including the header information and zero-padding bits. Then, the latency constraint is given by
%\begin{equation}
%	\frac{R_s+M L_o}{R_{\rm u}}\leq T,
%\end{equation}
%where  $B$ is the uplink bandwidth. Satisfying a stringent latency requirement generally requires a smaller source rate $R_s$, which leads to larger quantization step scale $s$ and hence larger compression loss.

\subsection{Over-the-Orbit Token Generation}
After receiving the uplink signals, the satellite performs demodulation and channel decoding to recover the transmitted packets. By collecting all decoded packets and performing depacketization, deinterleaving, and source decoding, the satellite obtains a distorted token representation, denoted by
\begin{equation}
	\tilde{\mathbf{Y}}=
	\left[
	\tilde{\mathbf{y}}_1,\tilde{\mathbf{y}}_2,\ldots,\tilde{\mathbf{y}}_{K}
	\right]^{\mathsf T}
	\in\mathbb{R}^{K\times D}.
\end{equation}
Due to the severe channel impairments in satellite communications, some packets may be decoded incorrectly or lost entirely. The lost tokens are replaced by an all-zero vector.

Equivalently, the distorted token representation can be  modeled as
\begin{equation}\label{packet}
	\tilde{\mathbf{Y}}=\mathbf{A}\odot \hat{\mathbf{Y}},
\end{equation}
where $\mathbf{A}\in\{0,1\}^{K\times D}$ is the error mask matrix indicating the positions of lost tokens and $\odot$ denotes the Hadamard product. Specifically, if the $n$-th token is lost, then the $n$-th row of $\mathbf{A}$ is an all-zero vector; otherwise, it is an all-one vector. Since the packet loss is determined by the channel decoding outcome, the error mask $\mathbf{A}$ can be obtained at the satellite through channel decoding and a cyclic redundancy check (CRC) \cite{Peterson1961CyclicCodes}. 

Based on the received signal $\tilde{\mathbf{Y}}$, we design the generation module to recover the corrupted tokens, yielding
\begin{equation}
	\bar{\mathbf{Y}}=G_{\bm{\Theta}}\!\left(\tilde{\mathbf{Y}},\mathbf{A}\right),
\end{equation}
\subsection{Downlink Forwarding Transmissions}

After the generation process, the satellite forwards the recovered token representation $\bar{\mathbf{Y}}$ to the destination ground device. 
 Let $R_{s,\mathrm{dl}}$ denote the average source bit length  for encoding $\bar{\mathbf{Y}}$. The coding process is the same as the uplink process. 

Different from the uplink, where limited transmit power and antenna gain make token losses critical, the downlink is considered to operate at a sufficiently high transmission rate with a low packet-error rate \cite{3gpp_tr38811}. Therefore, the downlink rate (bits/s) is  modeled by the Shannon capacity
\begin{equation}\label{eq:downlink_rate}
	R_{\rm d} = B_{\mathrm{dl}}\log_2(1+\gamma_{\mathrm{dl}}),
\end{equation}
where $B_{\mathrm{dl}}$ denotes the downlink bandwidth and $\gamma_{\mathrm{dl}}$ is the received downlink signal-to-noise ratio. 

After receiving the downlink bitstream, the destination ground user performs entropy decoding and token reconstruction to recover the token representation. Let $\breve{\mathbf{Y}}\in\mathbb{R}^{K\times D}$ denote the finally received token representation at the user side. The reconstructed image is then obtained through the semantic decoder $D_{\bm{\psi}}(\cdot)$ as
\begin{equation}
	\hat{\mathbf{X}} = D_{\bm{\psi}}(\breve{\mathbf{Y}}).
\end{equation}

\subsection{Problem Formulation}

Designing the satellite GF system gives rise to two major problems.

\subsubsection{Token Generation as an Inverse Problem}

According to \eqref{eq:additive_token_noise} and \eqref{packet}, the token representation received at the satellite can be modeled as
\begin{equation}\label{inverse}
	\tilde{\mathbf{Y}}=
	\underbrace{\mathbf{A}}_{\scriptstyle \text{packet loss}}
	\odot
	\left(
	\mathbf{Y}+
	\underbrace{\mathbf{Q}}_{\scriptstyle \text{compression loss}}
	\right).
\end{equation}
Accordingly, the received tokens are corrupted by channel-induced packet loss and source-compression loss. Given $\mathbf{A}$ and the covariance matrix of $\mathbf{Q}$, recovering $\mathbf{Y}$ from $\tilde{\mathbf{Y}}$ becomes an ill-posed inverse problem \cite{GuerreroColon2008ImageRestoration}. Conventional methods, such as Bayesian-based restoration \cite{GuerreroColon2008ImageRestoration}, require an accurate prior model of the token distribution; otherwise, their performance can degrade substantially. This motivates the use of diffusion models, which learn the token distribution directly from training data. However, conventional diffusion-based methods, such as diffusion posterior sampling (DPS) \cite{chung2023diffusion}, are primarily designed for image-domain inverse problems such as inpainting and do not account for the specific characteristics of token representations in satellite systems.

\subsubsection{On-Board Deployment Under Energy Limitation}

The on-board deployment of generation models remains challenging for satellite systems. Satellites are subject to stringent constraints on computational resources due to limited solar energy harvesting. Therefore, the design of the generation models must balance  performance, computational complexity, and energy resources. This issue will be discussed in detail in Section V.

\section{Channel-distortion-aware Diffusion Theory}

In this section, we derive a channel-distortion-aware diffusion theory by solving the  inverse problem in \eqref{inverse}. First, we introduce the corresponding SDE  and then present the solution algorithm. Then, we present the performance analysis by benchmarking against the DF scheme.

\subsection{Channel Posterior-based SDE}

The core idea of diffusion theory is to define a forward process that gradually perturbs a clean data into  isotropic Gaussian noise, and then learn a reverse process that progressively removes the  noise. To solve the inverse problem in \eqref{inverse},  the reverse process is further developed  to generate the desired data not only from noise but also conditioned on the received measurement $\tilde{\mathbf{Y}}$.

Specifically, the forward process can be defined by It\^{o} SDE \cite{Song2021ScoreBasedSDE}: 
\begin{equation}	\label{eq:forward_sde}
	d\mathbf{Y}_t=-\frac{\alpha(t)}{2}\mathbf{Y}_tdt+\sqrt{\alpha(t)}d\mathbf{W}_t,
\end{equation}
where $\mathbf{W}_t \in \mathbb{R}^{K\times D}$ is the standard Wiener process \cite{Song2021ScoreBasedSDE} and $\alpha(t)>0$ is the noise schedule of the process. When $t=0$, $\mathbf{Y}_{0}=\mathbf{Y} \sim p(\mathbf{Y})$. When $t=T$ and $T$ is sufficiently large,  the data is overwhelmed by Gaussian noise, i.e.,  $\mathbf{Y}_T\sim \mathcal{N}(\mathbf{0},\mathbf{I})$.

Then, the reverse process iteratively generates the clean data from both the Gaussian noise and $\tilde{\mathbf{Y}}$, which follows the  It\^{o} reverse SDE \cite{Song2021ScoreBasedSDE}:
\begin{align}	\label{eq:reverse_sde}
	d\mathbf{Y}_t=&\left[-\frac{\alpha(t)}{2}\mathbf{Y}_t-\alpha(t)\nabla_{\mathbf{Y}_t} \log p_t(\mathbf{Y}_t\mid \tilde{\mathbf{Y}};\mathbf{A}, \mathbf{\Sigma}_q)\right]dt \nonumber \\
		&+\sqrt{\alpha(t)}d\bar{\mathbf{W}}_t,
\end{align}
where $\bar{\mathbf{W}}_t \in \mathbb{R}^{K\times D}$ is the standard Wiener process \cite{Song2021ScoreBasedSDE}   and $p_t(\mathbf{Y}_t|\tilde{\mathbf{Y}};\mathbf{A}, \mathbf{\Sigma}_q)$ is the probability density function (PDF) of $\mathbf{Y}_t$ conditioned on  $\tilde{\mathbf{Y}}$ with $\mathbf{A}$ and $ \mathbf{\Sigma}_q$.  Note that $\nabla_{\mathbf{Y}_t} \log p_t(\mathbf{Y}_t\mid \tilde{\mathbf{Y}};\mathbf{A}, \mathbf{\Sigma}_q)$ is known as the posterior-based \emph{score function}, which points toward the direction  how a noisy sample can be closer to the data sample. By Bayesian rule \cite{rohatgi2015introduction},  the score function can be written as 
\begin{align} \label{score}
\nabla_{\mathbf{Y}_t} \log p_t(\mathbf{Y}_t|\tilde{\mathbf{Y}};\mathbf{A}, \mathbf{\Sigma}_q) =& \nabla_{\mathbf{Y}_t} \log p_t(\mathbf{Y}_t) \nonumber \\
&+ \nabla_{\mathbf{Y}_t} \log p_t(\tilde{\mathbf{Y}}\mid \mathbf{Y}_t;\mathbf{A}, \mathbf{\Sigma}_q)
\end{align}
The first term is the prior score, which can be approximated by a pre-trained score network $\nabla_{\mathbf{Y}_t} \log p_t(\mathbf{Y}_t)\approx \mathbf{S}_{\bm{\Theta}}(\mathbf{Y}_t,t)$ with parameter $\bm{\Theta}$.
The main challenge is  to characterize the PDF
$p_t(\tilde{\mathbf{Y}}\mid \mathbf{Y}_t;\mathbf{A}, \mathbf{\Sigma}_q)$.
For notation simplicity, we omit $\{\mathbf{A}, \mathbf{\Sigma}_q\}$ in the PDFs. By Bayes' rule,
$p_t(\tilde{\mathbf{Y}}\mid \mathbf{Y}_t)$ can be factorized as
\begin{align} \label{bayesian}
	p_t(\tilde{\mathbf{Y}}\mid \mathbf{Y}_t)
	&= \int p_t(\tilde{\mathbf{Y}}, \mathbf{Y}\mid \mathbf{Y}_t)\, d\mathbf{Y} \nonumber \\
	&= \int p_t(\mathbf{Y}\mid \mathbf{Y}_t)\, p(\tilde{\mathbf{Y}}\mid \mathbf{Y})\, d\mathbf{Y}.
\end{align}

The conditional PDF $p_t(\mathbf{Y}\mid \mathbf{Y}_t)$ is intractable. Hence, we adopt a pseudo-inverse Gaussian approximation \cite{Song2023PiGDM}:
\begin{align} \label{gaussian}
	p_t(\mathbf{Y}\mid \mathbf{Y}_t)
	\approx
	\mathcal{N}(\check{\mathbf{Y}}_t, r_t^2\mathbf{I}_{K\times D}),
\end{align}
where
$
	\check{\mathbf{Y}}_t \triangleq \mathbb{E}(\mathbf{Y}\mid \mathbf{Y}_t),
$
and $r_t>0$ is a time-dependent parameter of the Gaussian approximation. Moreover, $p(\tilde{\mathbf{Y}}\mid \mathbf{Y})$ can be directly obtained from the inverse model in \eqref{inverse}. Based on these models, we have the following proposition.

\begin{proposition}[Channel Posterior-based Score Function]
	\label{prop:posterior_score_spacediffusion}
	For the inverse problem in \eqref{inverse} with given $\{\mathbf{A}, \mathbf{\Sigma}_q\}$, the posterior-based score function can be approximated as
	\begin{align}
		\nabla_{\mathbf{Y}_t} \log p_t(\mathbf{Y}_t\mid\tilde{\mathbf{Y}})
		\approx
		\mathbf{S}_{\bm{\Theta}}(\mathbf{Y}_t,t)
		+
		\underbrace{\operatorname{unvec}\!\left(\mathbf{J}_t^{\top}\operatorname{vec}(\mathbf{G}_t)\right)}_{\triangleq \mathcal{L}_t},
	\end{align}
	where $\operatorname{vec}(\cdot)$ and $\operatorname{unvec}(\cdot)$ denote the vectorization and unvectorization operations, respectively,
	$
		\mathbf{J}_t
		\triangleq
		\frac{\partial \operatorname{vec}(\check{\mathbf{Y}}_t)}
		{\partial \operatorname{vec}(\mathbf{Y}_t)},
	$ denotes the vectorized Jacobian matrix of $\check{\mathbf{Y}}_t$ w.r.t. $\mathbf{Y}_t$,
	and
	\begin{align}
		\mathbf{G}_t
		\triangleq
		\mathbf{A}\odot
		\left(
		(\tilde{\mathbf{Y}}-\check{\mathbf{Y}}_t)
		(r_t^2\mathbf{I}_D+\mathbf{\Sigma}_q)^{-1}
		\right).
	\end{align}
\end{proposition}

\begin{IEEEproof}
	Please see Appendix A.
\end{IEEEproof}

Based on the score function in Proposition \ref{prop:posterior_score_spacediffusion}, the reverse SDE in \eqref{eq:reverse_sde} can be solved by generic stochastic samplers, such as Euler--Maruyama algorithm \cite{Higham2001EulerMaruyama}. However, these methods usually require many steps with repeated random perturbations, resulting in high computational cost and unstable reconstruction quality.

\subsection{Channel-distortion-aware DDIM Algorithm}

In this subsection, we resort to an efficient discretization algorithm, namely, the DDIM sampler \cite{Song2021DDIM}, to solve the SDEs in \eqref{eq:forward_sde} and \eqref{eq:reverse_sde}. Compared with generic stochastic samplers, DDIM provides more stable reconstruction quality while requiring only a small number of sampling steps. We discretize time into a total of $T$ steps, where $t\in\{0,1,\ldots,T\}$ denotes the discrete step index. We then present the forward and reverse processes in detail.

\subsubsection{ \textbf{Forward process}}
After discretization, the forward process at step $t$ is given by
\begin{equation}\label{eq:ddim_forward}
	\mathbf{Y}_{t}
	=
	\sqrt{\bar{\alpha}_{t}}\,\mathbf{Y}_0
	+
	\sqrt{1-\bar{\alpha}_{t}}\,\boldsymbol{\epsilon}_t,
	\qquad
	\boldsymbol{\epsilon}_t\sim\mathcal{N}(\mathbf{0},\mathbf{I}_{K\times D}),
\end{equation}
where $\bar{\alpha}_{t}\in(0,1]$ is a predesigned cumulative noise schedule \cite{Song2021DDIM}.

\begin{figure*}[t]
	\normalsize
	\setlength{\abovecaptionskip}{4pt}
	\setlength{\belowcaptionskip}{-0.1cm}
	\centering
	\subfigure[]{
		\includegraphics[width=0.46\linewidth]{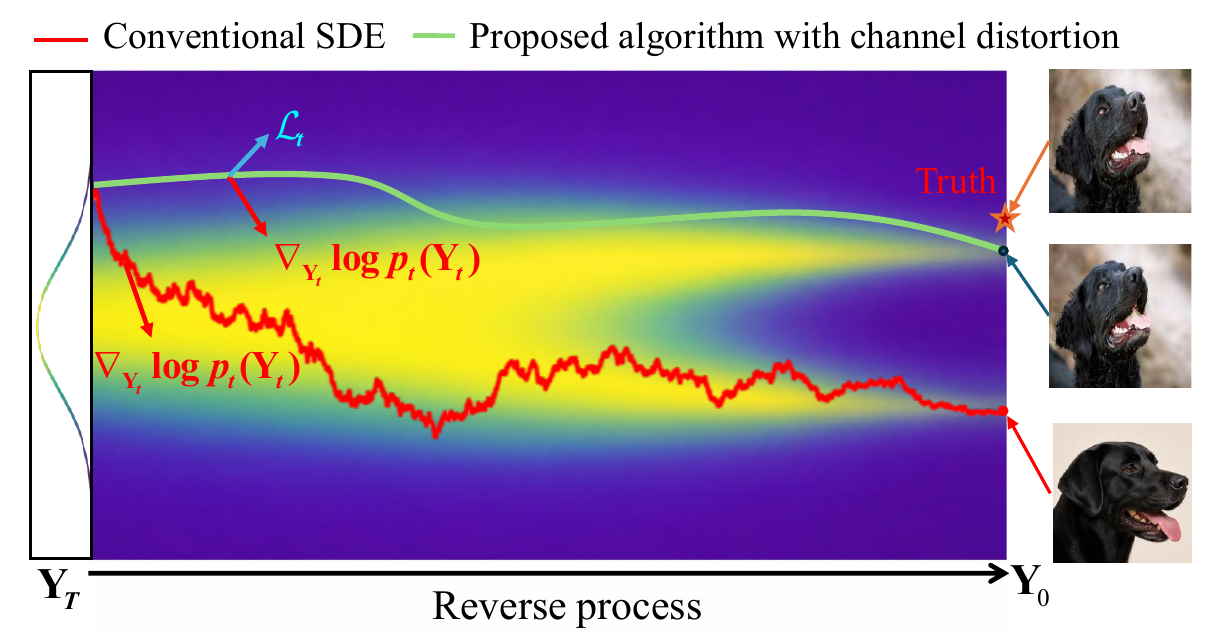}
	}
	% \hfill
	\subfigure[]{
		\includegraphics[width=0.46\linewidth]{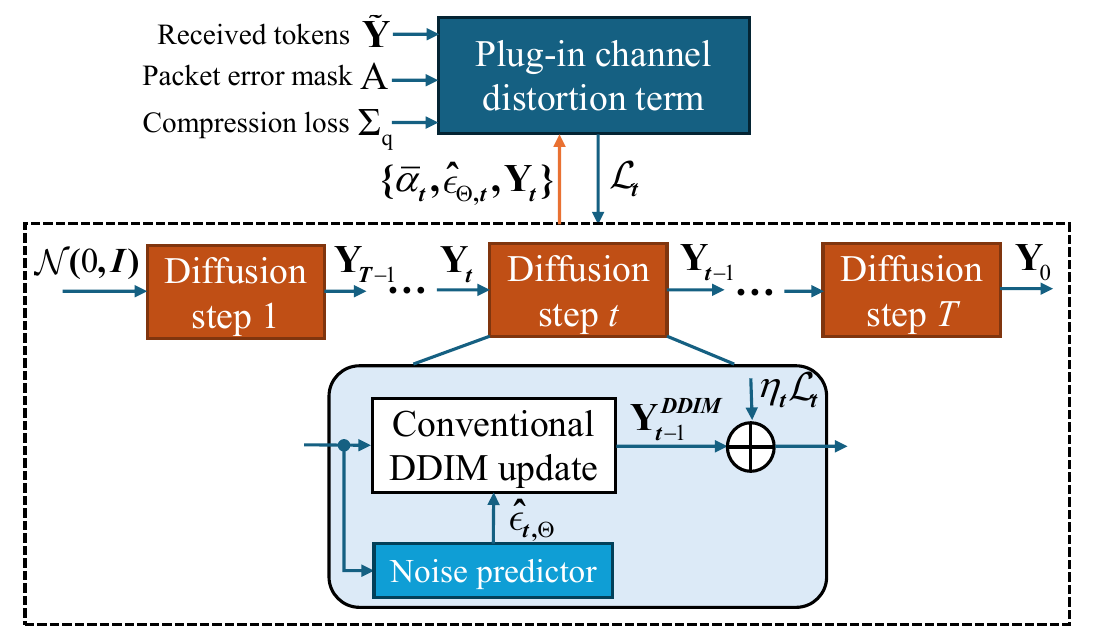}
	}
	%\caption{fig2}
	\captionsetup{justification=justified}
	\caption{Illustration of the proposed posterior-based Diffusion theory. (a) Comparison of the reverse process with and without the plug-in term $\mathcal{L}_t$. (b) Computational flow of the reverse process of the plug-and-play DDIM algorithm. }
	\label{DDIM_figure}
	\vspace{-5pt}
\end{figure*}

\subsubsection{ \textbf{Reverse process with plug-in term}} In the reverse process, the goal is to predict the data sample $\mathbf{Y}_{t-1}$ from the noisy sample $\mathbf{Y}_t$ and the measurement $\tilde{\mathbf{Y}}$. Essentially, this is equivalent to learning the conditional probability $p_{\bm{\Theta}}(\mathbf{Y}_{t-1}\mid \mathbf{Y}_t, \tilde{\mathbf{Y}})$, where $\bm{\Theta}$ denotes the deep neural network (DNN) parameters. According to the conventional DDIM update \cite{Song2021DDIM}, we propose to model $p_{\bm{\Theta}}(\mathbf{Y}_{t-1}\mid \mathbf{Y}_t, \tilde{\mathbf{Y}})$ as
	\begin{align} \label{eq:ddim_reverse}
		\mathcal{N}\!\left(
		\mathbf{Y}_{t-1};\,
		\sqrt{\bar{\alpha}_{t-1}}\widehat{\mathbf{Y}}^{\mathrm{post}}_t+\sqrt{1-\bar{\alpha}_{t-1}-\sigma_t^2}\,\hat{\boldsymbol{\epsilon}}_{\bm{\Theta},t},\,
		\sigma_t^2\mathbf{I}_{K\times D}
		\right),
	\end{align}
	where
	$
		\widehat{\mathbf{Y}}^{\mathrm{post}}_t \triangleq \mathbb{E}(\mathbf{Y}\mid \mathbf{Y}_t,\tilde{\mathbf{Y}})
		$ is the posterior mean, and $\sigma_t^2$ is a variance parameter that controls the randomness of the reverse process. It is usually set close to zero for stable sampling. Here, $\hat{\boldsymbol{\epsilon}}_{\bm{\Theta},t}( \mathbf{Y}_t, t)$ is the noise predictor. It is implemented by a U-Net \cite{rombach2022high} and trained by minimizing the following MSE loss:
\begin{equation}
	\min_{\bm{\Theta}}\;
	\mathbb{E}_{\mathbf{Y}_0,t,\boldsymbol{\epsilon}}
	\left[
	\left\|
	\boldsymbol{\epsilon}_t
	-
	\hat{\boldsymbol{\epsilon}}_{\bm{\Theta}}(\mathbf{Y}_t,t)
	\right\|_2^2
	\right],
\end{equation}
where $\mathbf{Y}_t$ is generated according to \eqref{eq:ddim_forward}. The noise predictor is related to the score function by
$		
	\mathbf{S}_{\bm{\Theta}}(\mathbf{Y}_t,t)
	\approx
	-\frac{\hat{\boldsymbol{\epsilon}}_{\bm{\Theta},t}}{\sqrt{1-\bar{\alpha}_t}}.
$

		In \eqref{eq:ddim_reverse}, the posterior mean can be calculated as
		\begin{align}
		\widehat{\mathbf{Y}}^{\mathrm{post}}_t
		&\overset{(a)}{=}
		\frac{1}{\sqrt{\bar{\alpha}_t}}
		\left(
		\mathbf{Y}_t
		+
		(1-\bar{\alpha}_t)
		\nabla_{\mathbf{Y}_t}\log p_t(\mathbf{Y}_t\mid\tilde{\mathbf{Y}})
		\right) \nonumber\\
		&\overset{(b)}{\approx}
		\frac{1}{\sqrt{\bar{\alpha}_t}}
		\left(
		\mathbf{Y}_t
		+
		(1-\bar{\alpha}_t)\mathbf{S}_{\bm{\Theta}}(\mathbf{Y}_t,t)
		+
		(1-\bar{\alpha}_t)\mathcal{L}_t
		\right) \nonumber\\
		&\overset{}{=}
		\frac{1}{\sqrt{\bar{\alpha}_t}}
		\left(
		\mathbf{Y}_t
		+
		(1-\bar{\alpha}_t)\mathbf{S}_{\bm{\Theta}}(\mathbf{Y}_t,t)
		\right) 
		+
		\frac{1-\bar{\alpha}_t}{\sqrt{\bar{\alpha}_t}}\mathcal{L}_t \nonumber\\
		&\overset{(c)}{=}
		\check{\mathbf{Y}}_t
		+
		\frac{1-\bar{\alpha}_t}{\sqrt{\bar{\alpha}_t}}\mathcal{L}_t.
		\label{eq:posterior_mean_decomp}
	\end{align}
		where $(a)$ follows by applying Tweedie's approach \cite{Song2021DDIM} into posterior mean calculations, $(b)$ follows from Proposition \ref{prop:posterior_score_spacediffusion}, and $(c)$ uses the 
	$
		\check{\mathbf{Y}}_t \triangleq \mathbb{E}(\mathbf{Y}\mid\mathbf{Y}_t)
		=
		\frac{1}{\sqrt{\bar{\alpha}_t}}
		\left(
		\mathbf{Y}_t
		+
		(1-\bar{\alpha}_t)\mathbf{S}_{\bm{\Theta}}(\mathbf{Y}_t,t)
		\right) 
	$ from Tweedie's formula \cite{Song2021DDIM}.
		Further, according to the score-noise relation, we have
	\begin{align}
		\check{\mathbf{Y}}_t
		\approx
		\frac{1}{\sqrt{\bar{\alpha}_t}}
		\left(
		\mathbf{Y}_t
		-
		\sqrt{1-\bar{\alpha}_t}\,\hat{\boldsymbol{\epsilon}}_{\bm{\Theta},t}
		\right).
		\label{eq:checkY_from_noise}
	\end{align}

% At each reverse step $t\rightarrow t-1$, we first compute the conventional DDIM update []
% \begin{equation}\label{eq:ddim_base}
% 	\mathbf{Y}^{\mathrm{ddim}}_{t-1}
% 	=
% 	\sqrt{\bar{\alpha}_{t-1}}
% 	\check{\mathbf{Y}}_t
% 	+
% 	\sqrt{1-\bar{\alpha}_{t-1}}\,\hat{\boldsymbol{\epsilon}}_{\bm{\Theta},t},
% \end{equation}
% where 
% \begin{align}
% 	\check{\mathbf{Y}}_t=\mathbb{E}(\mathbf{Y}\mid \mathbf{Y}_t)\approx\frac{1}{\sqrt{\bar{\alpha}_t}}\left(\mathbf{Y}_t-\sqrt{1-\bar{\alpha}_{t}}\,\hat{\boldsymbol{\epsilon}}_{\bm{\Theta},t}\right).
% \end{align}
% Here,  $\hat{\boldsymbol{\epsilon}}_{\bm{\Theta},t}$ is the noise prediction function, which is related to the score function by
% $		
% 	\mathbf{S}_{\bm{\Theta}}(\mathbf{Y}_t,t)
% 	\approx
% 	-\frac{\hat{\boldsymbol{\epsilon}}_{\bm{\Theta},t}}{\sqrt{1-\bar{\alpha}_t}}.
% $

\noindent According to the reverse conditional model in \eqref{eq:ddim_reverse}, for each reverse step $t\rightarrow t-1$, $\mathbf{Y}_{t-1}$ can be written as
\begin{align} \label{eq:ddim_reverse_sample2}
	\mathbf{Y}_{t-1}
	&=
	\sqrt{\bar{\alpha}_{t-1}}\widehat{\mathbf{Y}}^{\mathrm{post}}_t
	+
	\sqrt{1-\bar{\alpha}_{t-1}-\sigma_t^2}\,\hat{\boldsymbol{\epsilon}}_{\bm{\Theta},t}
	+
	\sigma_t\mathbf{z}_t,
\end{align}
where $\mathbf{z}_t\sim\mathcal{N}(\mathbf{0},\mathbf{I}_{K\times D})$. Substituting \eqref{eq:posterior_mean_decomp} into \eqref{eq:ddim_reverse_sample2} yields
\begin{equation}\label{eq:plugin_update}
	\mathbf{Y}_{t-1}
	=
	\underbrace{\mathbf{Y}^{\mathrm{ddim}}_{t-1}}_{\scriptstyle \text{Direction to }p(\mathbf{Y})}
	+
	\underbrace{\eta_t\mathcal{L}_{t}}_{\scriptstyle \text{Direction based on channel distortion}},
\end{equation}
where $\eta_t=\eta \sqrt{\bar{\alpha}_{t-1}}(1-\bar{\alpha}_t)/\sqrt{\bar{\alpha}_t}$ and
\[
		\mathbf{Y}^{\mathrm{ddim}}_{t-1}
		=
		\sqrt{\bar{\alpha}_{t-1}}\check{\mathbf{Y}}_t
		+
		\sqrt{1-\bar{\alpha}_{t-1}-\sigma_t^2}\,\hat{\boldsymbol{\epsilon}}_{\bm{\Theta},t}
		+
		\sigma_t\mathbf{z}_t.
		\]
Here, we introduce a parameter $\eta>0$ to control the strength of the channel distortion correction. From \eqref{eq:plugin_update}, the first term is exactly the conventional DDIM step \cite{Song2021DDIM}, which drives the sample toward the data distribution, whereas the second term is a channel-aware correction term determined by the received signal $\tilde{\mathbf{Y}}$ and the channel condition.  
Note that in the proposed DDIM algorithm, the posterior correction $\mathcal{L}_t$ is directly inserted into the reverse update in \eqref{eq:plugin_update} without additional training. 

% \begin{algorithm}[t]
% 	\caption{Plug-and-Play DDIM Reverse Algorithm}
% 	\label{alg:pnp_ddim_reverse}
% 	\begin{algorithmic}[1]
% 		\Require Noisy token $\mathbf{Y}_{T}$, measurement $\tilde{\mathbf{Y}}$, noise predictor $\hat{\boldsymbol{\epsilon}}_{\bm{\Theta}}$, posterior correction $\mathcal{L}_{t}$, noise schedule $\{\bar{\alpha}_{t}\}_{t=0}^{T}$, step size $\eta$, and variance schedule $\{\sigma_t\}_{t=1}^{T}$
% 		\Ensure Reconstructed token $\mathbf{Y}_{0}$
% 		\For{$t=T,T-1,\ldots,1$}
% 		\State Compute the prior mean
% 		\[
% 		\check{\mathbf{Y}}_t
% 		=
% 		\frac{1}{\sqrt{\bar{\alpha}_t}}
% 		\left(
% 		\mathbf{Y}_t
% 		-
% 		\sqrt{1-\bar{\alpha}_t}\,\hat{\boldsymbol{\epsilon}}_{\bm{\Theta},t}
% 		\right)
% 		\]
% 		\State Compute  $\mathcal{L}_t$ according to Proposition \ref{prop:posterior_score_spacediffusion}
% 		\State Compute the conventional DDIM step
% 		\[
% 		\mathbf{Y}^{\mathrm{ddim}}_{t-1}
% 		=
% 		\sqrt{\bar{\alpha}_{t-1}}\check{\mathbf{Y}}_t
% 		+
% 		\sqrt{1-\bar{\alpha}_{t-1}-\sigma_t^2}\,\hat{\boldsymbol{\epsilon}}_{\bm{\Theta},t}
% 		+
% 		\sigma_t\mathbf{z}_t
% 		\]
% 		\State Update 
% 		$
% 		\mathbf{Y}_{t-1}
% 		$ according to \eqref{eq:plugin_update}.
% 		\EndFor
% 	\end{algorithmic}
% \end{algorithm}

\subsection{Performance Analysis}
In this subsection, we analyze the performance of the proposed  DDIM algorithm as the  reverse  step index $t$ decreases from $T$ to $0$. 
Denote
$
	\delta_t
	\triangleq
	\eta\frac{\sqrt{\bar{\alpha}_{t-1}}(1-\bar{\alpha}_t)}{\sqrt{\bar{\alpha}_t}}.
$
Then, we calculate the reconstruction error at each reverse step:
\begin{equation}
	\mathcal{E}_{t-1}
	\triangleq
	\mathbb{E}\!\left[\left\|\mathbf{Y}_{t-1}-\mathbf{Y}\right\|_F^2\right].
\end{equation}
The expectation is taken over the diffusion noise, image token $\mathbf{Y}$, the error mask $\mathbf{A}$ and the compression noise $\mathbf{Q}$. The error mask $\mathbf{A}$ is caused by channel impairments and is determined by packet decoding. Specifically, if the $n$-th packet is decoded successfully, the corresponding mask row is all-one; otherwise, it is all-zero. Thus, we model
\begin{equation}
	\mathbf{A}
	=
	\mathbf{a}\mathbf{1}_D^{\mathsf T},
	\qquad
	a_n\overset{\mathrm{i.i.d.}}{\sim}\mathrm{Bernoulli}(1-\rho),
\end{equation}
where $\rho=\mathbb{E}\{\rho_m\}$ denotes the average packet error probability. 

\begin{approximation}[Denoising error model]\label{assump:denoising_error}
	At reverse step $t$, the noise predictor is approximated as:
	\begin{equation}
		\hat{\boldsymbol{\epsilon}}_{\bm{\Theta},t}
		\approx
		\boldsymbol{\epsilon}_t+\mathbf{U}_t,
	\end{equation}
	where $\mathbf{U}_t\in\mathbb{R}^{K\times D}$ denotes the training-induced denoising error independent of $\boldsymbol{\epsilon}_t$. The entries of $\mathbf{U}_t$ are i.i.d. Gaussian random variables with zero mean and fixed variance $\sigma_u^2$, which can be estimated by offline training.
	
\end{approximation}

 Then, we have the following result.

\begin{proposition}[Approximate Error Characterization]\label{thm:reconstruction_error_upper_bound}
		Under Approximation \ref{assump:denoising_error} and pseudo-inverse approximation in \eqref{gaussian}, the deterministic DDIM setting $\sigma_t=0$, and the Jacobian bound $\|\mathbf{J}_t\|_2\le \kappa$, the error at each reverse step $t-1$ can be approximately upper bounded by
	\begin{equation}\label{eq:reconstruction_error_upper_bound}
		\begin{aligned}
				\mathcal{E}_{t-1}
				& \lesssim
				\left(
				\sqrt{a_t^{\mathrm{ddim}}}
				+
				\delta_t\kappa \sqrt{\chi_t}
				\right)^2
				\triangleq
				\mathcal{E}_{t-1}^{\mathrm{up}}.
		\end{aligned}
	\end{equation}
	where
	\begin{equation}
		\mathbf{B}_t
		\triangleq
		(r_t^2\mathbf{I}_D+\mathbf{\Sigma}_q)^{-1},
	\end{equation}
		\begin{equation}
			\chi_t
			\triangleq
					(1-\rho)K
				\operatorname{tr}\!\left(
				\mathbf{B}_t
				\left(\mathbf{\Sigma}_q+\frac{1-\bar{\alpha}_t}{\bar{\alpha}_t}\sigma_u^2\mathbf{I}_D\right)
				\mathbf{B}_t
				\right),
		\end{equation}
	and
		\begin{equation}
			\begin{aligned}
				a_t^{\mathrm{ddim}}
				&=
					(\sqrt{\bar{\alpha}_{t-1}}-1)^2
						K\!\left(\|\boldsymbol{\mu}\|_2^2+\operatorname{tr}(\boldsymbol{\Sigma})\right)\\
						&\quad+
						(1-\bar{\alpha}_{t-1})KD\\
						&\quad+
						KD\sigma_u^2
					\left(
					\sqrt{1-\bar{\alpha}_{t-1}}
					-
					\sqrt{\frac{\bar{\alpha}_{t-1}(1-\bar{\alpha}_t)}{\bar{\alpha}_t}}
					\right)^2.
				\end{aligned}
				\label{eq:ddim_error_with_training_noise}
			\end{equation}
	Here, $\mathcal{E}_{T}=\mathbb{E}\left[\|\mathbf{Y}_T-\mathbf{Y}\|_F^2\right]=K\!\left(\|\boldsymbol{\mu}\|_2^2+\operatorname{tr}(\boldsymbol{\Sigma})\right)+
						KD$ with $\mathbf{Y}_T \sim \mathcal{N}(0,\bm{I})$.
	\end{proposition}

\begin{IEEEproof}
	Please see Appendix B.
\end{IEEEproof}

\begin{remark}[Progressive Denoising Trend] \label{remark} 
The reverse-time SDE in \eqref{eq:reverse_sde} progressively transports noisy samples toward the data distribution as the step index $t$ decreases \cite{Song2021ScoreBasedSDE}. Accordingly, $\bar{\alpha}_{t-1}$ approaches one and the prescribed diffusion variance decreases. When the diffusion model is well trained, such that $\sigma_u^2$ is sufficiently small, and the posterior correction remains bounded, these schedule-dependent contractions dominate the error bound. Consequently, $\mathcal{E}_{t-1}^{\mathrm{up}}$ exhibits an overall decreasing trend under the considered noise schedule, consistent with the progressive denoising behavior of the reverse SDE.

\end{remark}

Based on the tractable upper bound in Proposition \ref{thm:reconstruction_error_upper_bound}, we can further analyze the performance of the plug-and-play DDIM by comparing it with the  conventional DF scheme for transmitting the corrupted tokens $\tilde{\mathbf{Y}}$. Here, we assume that the compression loss for downlink transmission is negligible. Then, we have the following result.

\begin{corollary}[Diffusion Activation Threshold]\label{cor:upper_bound_below_mse} 
	For the DF scheme,  the corrupted tokens $\tilde{\mathbf{Y}}$ are transmitted to the receiver without generation process. Hence, the mean-squared error is calculated by
	\begin{equation}\label{eq:direct_mse}
		\begin{aligned}
			\mathrm{MSE}
			&=
			\mathbb{E}\!\left[\left\|\tilde{\mathbf{Y}}-\mathbf{Y}\right\|_F^2\right] \\
			&=
				K\!\left[
			\rho\left(\|\boldsymbol{\mu}\|_2^2+\operatorname{tr}(\boldsymbol{\Sigma})\right)
			+
			(1-\rho)\operatorname{tr}(\mathbf{\Sigma}_q)
			\right].
			\end{aligned}
			\end{equation}
	The reverse step predicted by \eqref{eq:reconstruction_error_upper_bound} to outperform the DF scheme is
	\begin{equation}
		t^{\mathrm{threshold}}
		=
		\max\{\,t\in\{0,\ldots,T-1\}:\mathcal{E}_{t}^{\mathrm{up}}<\mathrm{MSE}\,\}.
	\end{equation}
		\end{corollary}

\begin{IEEEproof}
	According to the inverse model in \eqref{inverse} , the MSE in \eqref{eq:direct_mse} can be easily calculated. Then, the threshold step $t^{\mathrm{threshold}}$ is derived by comparing the upper bound $\mathcal{E}_{t}^{\mathrm{up}}$ in \eqref{eq:reconstruction_error_upper_bound} with the MSE.
\end{IEEEproof}

\begin{figure}[t]
	\normalsize
	\setlength{\abovecaptionskip}{4pt}
	\setlength{\belowcaptionskip}{-0.1cm}
	\centering
	\includegraphics[width=0.9\linewidth]{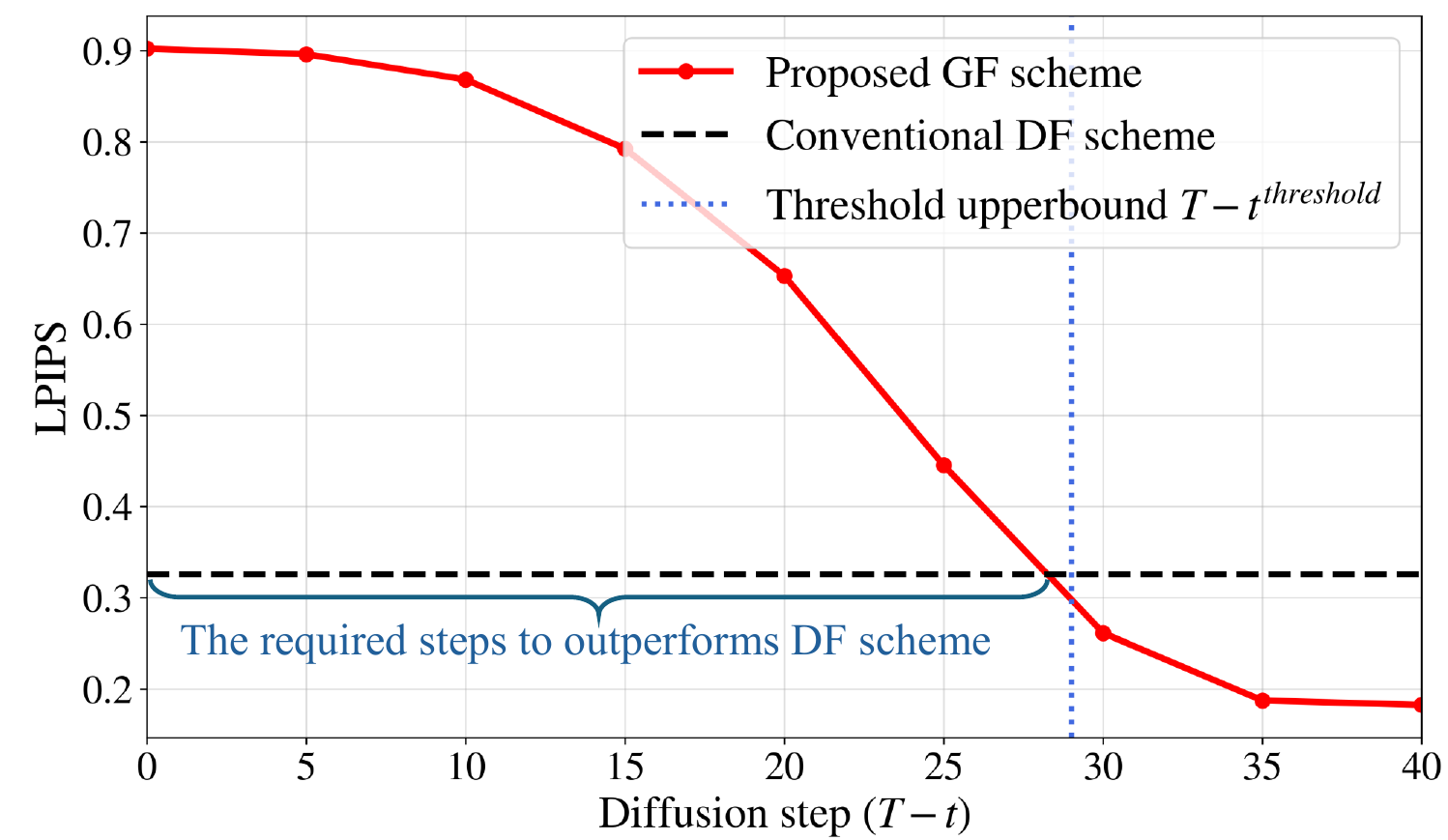}
	%\caption{fig2}
	\captionsetup{justification=justified}
	\caption{LPIPS score versus the number of reverse diffusion steps. The total step $T$, the average packet error  $\rho$, and the quantization step  $s$ are set as  $40$, $0.1$, and $1$, respectively.}
	\label{theorem_figure}
\end{figure} 

\noindent\textbf{Effect on E2E Image Quality.}
Note that the above analysis is conducted in the token domain. To examine how the recovered tokens affect E2E image quality, we feed the intermediate token representation at each reverse step into the image recovery module and obtain
\begin{equation}
	\hat{\mathbf{X}}_t
	=
	D_{\bm{\psi}}(\mathbf{Y}_t).
\end{equation}
Then, we compare $\hat{\mathbf{X}}_t$ with the original image $\mathbf{X}$ using the Learned Perceptual Image Patch Similarity (LPIPS) metric \cite{Zhang2018LPIPS}.

As shown in Fig. \ref{theorem_figure}, the LPIPS score decreases as the number of  diffusion steps increases.  Although Corollary \ref{cor:upper_bound_below_mse} is derived from the token-domain squared Frobenius norm, it still provides an effective estimate of the number of  diffusion steps required to outperform the DF scheme in terms of  image quality. The reason is that the learned tokens are trained to preserve  perceptually relevant information of image, reducing the reconstruction error in the token domain can generally improve the preservation of visual quality.

\section{Optimal Diffusion Steps on Satellite}  

This section optimizes the number of reverse diffusion steps that can be executed by the satellite under orbital energy limitations. First, we model the harvested solar energy and the energy consumed by diffusion computation and downlink forwarding. Then, we  derive the optimal number of diffusion steps.

\subsection{Energy Harvesting Model}

\subsubsection{Solar Energy Harvesting}
Denote the physical orbital time by $\tau$ to avoid confusion with the reverse diffusion index $t$. When $\tau=0$, the satellite is at the point that has
the greatest distance to the Sun. The harvested solar power is modeled as \cite{yang2016towards}
\begin{equation}
	P_{\rm h}(\tau)=
	\begin{cases}
		0, & \tau\in\Omega_{\rm ecl},\\[1mm]
		P_{max}\cos\theta(\tau), & \text{otherwise}.
	\end{cases}
\end{equation}
where $\Omega_{\rm ecl}$ denotes the eclipse interval, $P_{max}$ is the peak solar power determined by the solar-panel area and energy-conversion efficiency, and $\theta(\tau)$ is the angle between the solar-panel normal and the Sun direction. With single-axis solar tracking \cite{yang2016towards}, $\theta(\tau)$ is given by
\begin{align}
	\theta(\tau)
	&=
	\arccos\!\left(
	\sqrt{1-\cos^2\beta \cos^2(\tau\omega)}
	\right),
\end{align}
where $\beta$ denotes the angle between the satellite orbital
plane and the sunlight direction, and $\omega$ is the orbital angular velocity of the satellite.

\subsubsection{Energy Consumption}
The energy consumption of the satellite comes from  three parts: (i) reverse diffusion computation, (ii)  downlink forwarding, and (iii) other satellite base systems, such as attitude control and heat management \cite{nasa_smallsat_power_2026}.

Let $E_{\rm step}$ denote the energy consumption of one reverse diffusion step. In general, it can be written as
\begin{equation}
	E_{\rm step}=P_{\rm step}\tau_{\rm step},
\end{equation}
where $P_{\rm step}$ and $\tau_{\rm step}$ denote the average execution power and execution time of one reverse diffusion step, respectively. 

 From the downlink rate in \eqref{eq:downlink_rate}, the downlink transmission time for forwarding the recovered tokens is
$
	T_{\rm dl}^{\rm tx}
	=
		\frac{R_{s,\mathrm{dl}}}{R_{\rm d}},
$
where $T_{\mathrm{dl}}$ denotes the maximum allowable downlink transmission duration.
Let $p_{\rm dl}$ denote the satellite downlink transmit power. The downlink transmission energy is then
\begin{equation}
	E_{\rm dl}
	=
	p_{\rm dl}T_{\rm dl}^{\rm tx}
	=
	p_{\rm dl}\frac{R_{s,\mathrm{dl}}}{R_{\rm d}}.
\end{equation}

Let $E_{\rm base}$ denote the energy consumption from other satellite systems.  The total energy consumption relevant to diffusion scheduling is
\begin{equation}
	E_{\rm cons}(N)
	=
	NE_{\rm step}+E_{\rm base}+E_{\rm dl},
\end{equation}
where $N$ denotes the number of reverse diffusion steps to be executed.

\subsection{Diffusion Step Optimization}

Let $\tau_0$ denote the start time of on-board diffusion after uplink decoding.  Let $\tau_{\rm end}$ denote the latest time by which the recovered tokens must be delivered to the destination. Since downlink forwarding requires $T_{\rm dl}^{\rm tx}$ seconds, the maximum executable time for diffusion is
\begin{equation}
	T_{\rm exe}
	=
	\left[\tau_{\rm end}-\tau_0-T_{\rm dl}^{\rm tx}\right]^+.
\end{equation}
During this interval, the harvested energy is
\begin{equation}
	E_{\rm h}^{\rm exe}
	=
	\int_{\tau_0}^{\tau_0+T_{\rm exe}}P_{\rm h}(\tau)d\tau .
\end{equation}
Let $B_0$ denote the battery energy at $\tau_0$. The energy available for scheduling, after battery saturation and reserve protection, is
\begin{equation}
	E_{\rm av}
	=
	\left[
	\min\!\left\{B_{\max},B_0+E_{\rm h}^{\rm exe}\right\}
	-B_{\min}-E_{\rm base}-E_{\rm dl}
	\right]^+,
\end{equation}
where $B_{\max}$ is the battery capacity and $B_{\min}$ is the reserved battery energy.

According to Corollary \ref{cor:upper_bound_below_mse}, the plug-and-play DDIM is predicted to outperform the decode-and-forward baseline when the reverse process reaches $t^{\mathrm{threshold}}$. Hence, the minimum number of required  diffusion steps is
$
	N_{\rm th}
	=
	T-t^{\mathrm{threshold}}$.
Following the progressive denoising trend in Remark \ref{remark}, our objective is to maximize the number of executable diffusion steps within the interval $[\tau_0,\tau_0+T_{\rm exe}]$, while ensuring that the executed steps reach $N_{\rm th}$:
\begin{subequations}\label{prob:max_diffusion_steps}
\begin{align}
	\quad
	\max_{N}\quad
	& N \\
	{\rm s.t.}\quad
	& N\geq N_{\rm th},\\
	& N\tau_{\rm step}\le T_{\rm exe},\\
	& NE_{\rm step}\le E_{\rm av}.
\end{align}
\end{subequations}
The first constraint enforces the minimum number of reverse diffusion steps. The second constraint is the executable-time constraint, and the third constraint is the energy-causality constraint after reserving the downlink transmission energy.

The problem has a closed-form solution:
\begin{equation}
	N^{\star}
	=
	\begin{cases}
		\displaystyle
		\min\!\left\{
		T,\left\lfloor\dfrac{T_{\rm exe}}{\tau_{\rm step}}\right\rfloor,
		\left\lfloor\dfrac{E_{\rm av}}{E_{\rm step}}\right\rfloor
		\right\},
		& \text{if } N_{\max}^{\rm exe}\ge N_{\rm th},\\[3mm]
		0, & \text{otherwise},
	\end{cases}
\end{equation}
where
\begin{equation}
	N_{\max}^{\rm exe}
	=
	\min\!\left\{
	T,\left\lfloor\frac{T_{\rm exe}}{\tau_{\rm step}}\right\rfloor,
	\left\lfloor\frac{E_{\rm av}}{E_{\rm step}}\right\rfloor
	\right\}.
\end{equation}
When $N_{\max}^{\rm exe}< N_{\rm th}$, the predicted threshold cannot be reached under the current energy and time constraints. In this case, the diffusion process is infeasible, $N^{\star}=0$, and the system degenerates into the DF scheme. 

\section{Experimental Results}

\begin{figure*}[!t]
	\normalsize
	\setlength{\abovecaptionskip}{4pt}
	\setlength{\belowcaptionskip}{-0.1cm}
	\centering
	\subfigure[Latency versus satellite altitude.]{
		\includegraphics[width=0.45\textwidth]{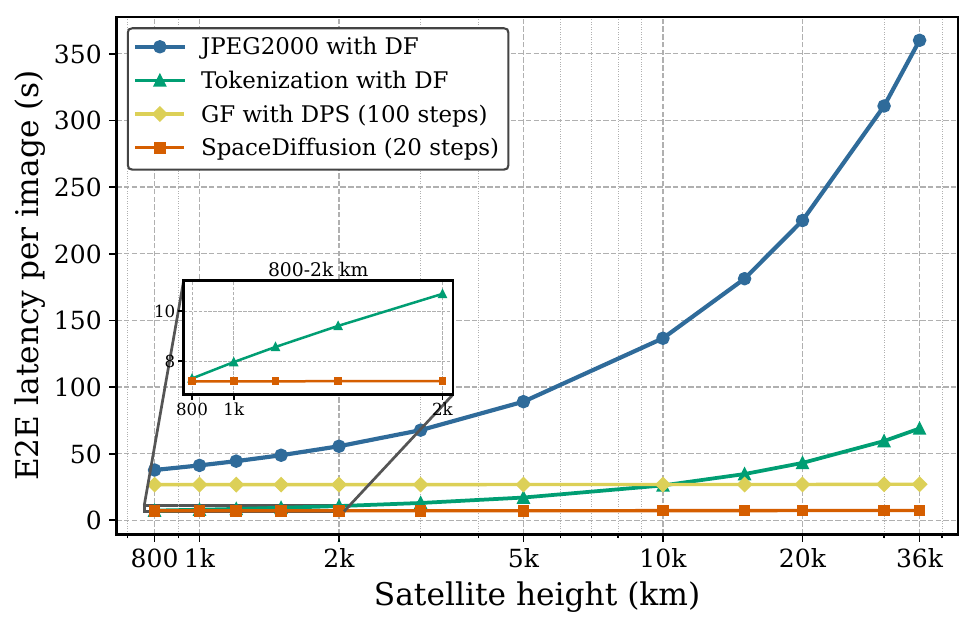}
		\label{fig:e2e_latency_vs_sat_height}
	}
	\hfill
	\subfigure[Latency breakdown (2000 km height).]{
		\includegraphics[width=0.45\textwidth]{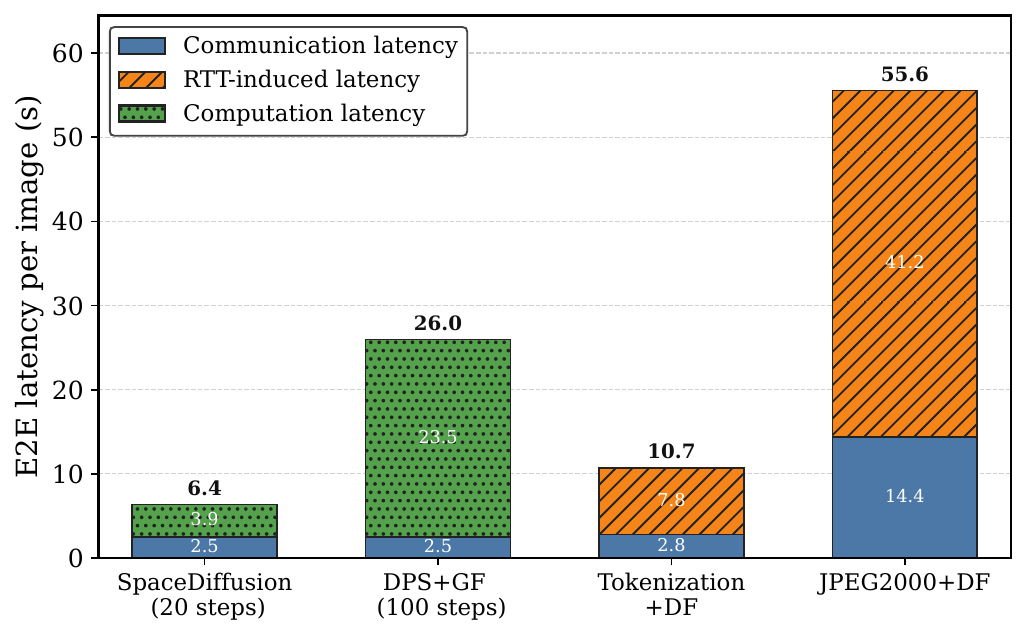}
		\label{fig:latency_breakdown_5k}
	}
	\captionsetup{justification=justified}
	\caption{E2E latency comparisons. The LPIPS for all the schemes are around $0.2$. }
	\label{fig:latency_comparisons}
	 \vspace{-7pt}
\end{figure*}
\subsection{Experimental Settings}

\subsubsection{\textbf{Model architecture}}
We implement the tokenization coders by using the first-stage convolutional autoencoder of a stable latent diffusion model \cite{rombach2022high}.  For the proposed plug-and-play DDIM algorithm, we employ the pretrained noise predictor from the stable latent diffusion model \cite{rombach2022high}. All neural-network modules are implemented in PyTorch \cite{paszke2019pytorch}. The NVIDIA A100 GPU is used as a terrestrial profiling platform for measuring the execution cost of the diffusion workload. This setting is reasonable because the recent orbital AI-computing platforms, such as Starcloud \cite{feilden2024train}, have adopted H100-class GPUs with higher AI throughput. 

\subsubsection{\textbf{Dataset}}
The experiments are conducted on the AFHQ dataset \cite{choi2020starganv2}, a well-known large-scale image dataset for high-quality image synthesis. AFHQ contains $15{,}000$ images with an original resolution of $512\times512$. We use the training dataset to estimate the token statistics $\{\boldsymbol{\mu},\boldsymbol{\Sigma}\}$, while all reported results are evaluated on the validation set. We further validate the visual performance on the Kodak dataset \cite{Kodak}.

\subsubsection{\textbf{Onboard deployment discussion}} The denoising UNet contains $859.52$ million parameters and occupies approximately $1.60$ GiB using FP16 weights. For a $512\times512$ image, each reverse step  with automatic mixed precision (AMP) takes approximately $167$ ms on the A100. These requirements are compatible with emerging high-performance orbital platforms such as Starcloud-1, while resource-constrained CubeSats would require model compression.
\subsubsection{\textbf{Parameter settings}}
The orbital parameters and computation-energy settings are chosen with reference to recent AI-computing satellite efforts, such as Starcloud and NASA microsatellite settings \cite{feilden2024train,nasa_smallsat_power_2026}, which are summarized in Table \ref{tab:experiment_parameters}. The source rate is controlled by the quantization step scale $s$ and the packet error rate is calculated according to the finite block length theory in \eqref{finite_eq}. The small scale fading follows the Rician distribution with factor being $1$, i.e., $g_m\sim \mathcal{CN}(\sqrt{\frac{1}{2}},\frac{1}{2})$. Downlink transmissions are assumed to be error-free.  The overhead of the packet header is around $20\%$ of the total source bits. The uplink transmission rate is around $18$ kb/s, while the downlink transmission rate can be hundred kHz according channel conditions. The downlink compression bit per pixel (bpp) is set as $0.4$.
\begin{table}[t]
	\centering
	\caption{Default simulation parameters.}
	\label{tab:experiment_parameters}
	\renewcommand{\arraystretch}{1.05}
	\footnotesize
	\begin{tabular}{p{0.39\columnwidth} p{0.50\columnwidth}}
		\hline
		\textbf{Parameter} & \textbf{Default setting} \\
		\hline
		DDIM parameters & $\sigma_t=0,\eta=0.8$ \\
		Packet length & $2048$ \\
		Uplink channel coding rate & $0.1$ \\
		Carrier frequency & $2$ GHz\\
		Bandwidth & $180$ kHz \\
		Uplink transmit power & $23$ dBm \\
		Downlink transmit power & $30$ dBm \\
		Device transmit/receive antenna gain & $0$ dBi \\
		Satellite receive antenna gain & $15$ dBi \\
		Satellite transmit antenna gain & $25$ dBi \\
		The elevation angle & $30^{\circ}$ \\
		Solar energy harvesting peak power & $800$ W \\
		Computational power & $200$ W \\
		Satellite base energy consumption & $200$ W \\
		\hline
	\end{tabular}
\end{table}

\subsubsection{\textbf{Benchmarking schemes}} Instead of using pixel-level MSE as the performance metric, we adopt the LPIPS metric to evaluate the perceptual quality of the received images.
The proposed SpaceDiffusion scheme is compared with the following schemes.
\begin{itemize}
	\item \textbf{JPEG2000+DF}: The source image is compressed by JPEG2000 \cite{1037027}, protected by the same channel code, decoded at the satellite, and forwarded to the destination without generative recovery. 
	\item \textbf{Tokenization+DF}: This scheme uses exactly the same
first-stage autoencoder, transform coding, quantization, entropy coding,
packetization, and channel coding as SpaceDiffusion. After channel
decoding, the corrupted latent tokens are directly forwarded without
on-board generative recovery. 
	
	\item \textbf{DPS+GF}: As a generative baseline, we employ the conventional DPS algorithm \cite{chung2023diffusion} for token generation at the satellite. Since DPS is originally used for image-domain inverse problem, we redesign it for token-domain generation without changing its main structure.  
\end{itemize}

\subsection{E2E Latency Comparisons}

First, we investigate the E2E latency of the proposed SpaceDiffusion scheme over different satellite heights. The elevation angle is set to $5^{\circ}$ to investigate the poor channel conditions of satellite communications. The packet error rate is set as $0.1$. The considered DF relay adopts single-process stop-and-wait retransmission protocol in NB-IoT NTN setting.
When an uplink packet is decoded, the satellite feeds back an ACK/NACK to the source, and the source retransmits the lost packet until successful decoding.  Therefore, the latency of the DF scheme explicitly includes both the extra transmission time caused by retransmissions and the ACK/NACK-induced RTT penalty.

Fig.~\ref{fig:latency_comparisons} compares the schemes at the same target quality of approximately $0.2$ LPIPS. SpaceDiffusion maintains a nearly constant latency of $6$--$7$ s/image across the considered satellite heights because it avoids uplink retransmissions and ACK/NACK feedback. In contrast, the DF baselines become increasingly slow as altitude increases, with JPEG2000+DF exceeding $300$ s/image at high altitude. SpaceDiffusion consequently achieves a $4$--$50\times$ latency reduction over the DF baselines and is also faster than DPS+GF, which requires $100$ rather than $20$ reverse steps.

The breakdown in Fig.~\ref{fig:latency_breakdown_5k} explains these gains. Generative recovery enables SpaceDiffusion to use about $6\times$ lower bpp than JPEG2000 at the same perceptual quality. It also replaces the $7.8$ s/image RTT component of Tokenization+DF with $3.9$ s/image of local diffusion computation, reducing latency by about $40\%$. Moreover, using $20$ rather than $100$ reverse steps reduces the DPS computation latency from $23.5$ to $3.9$ s/image. It is worth noticing that the computation latency of the proposed scheme can be further reduced by using more lightweight models and parallel image processing.

\subsection{Reconstruction Performance Comparisons}

In this subsection, we investigate the performance of the proposed SpaceDiffusion scheme compared with the benchmarking schemes.
The quantization step is set as $0.8$, the bpp is around $0.1$, and the bandwidth ratio is $0.49$ for all schemes. The satellite height is set as $600$ km. The solar power is enough to support the diffusion computation.
\begin{figure}[t]
	\normalsize
	\setlength{\abovecaptionskip}{4pt}
	\setlength{\belowcaptionskip}{-0.1cm}
	\centering
	\includegraphics[width=0.85\linewidth]{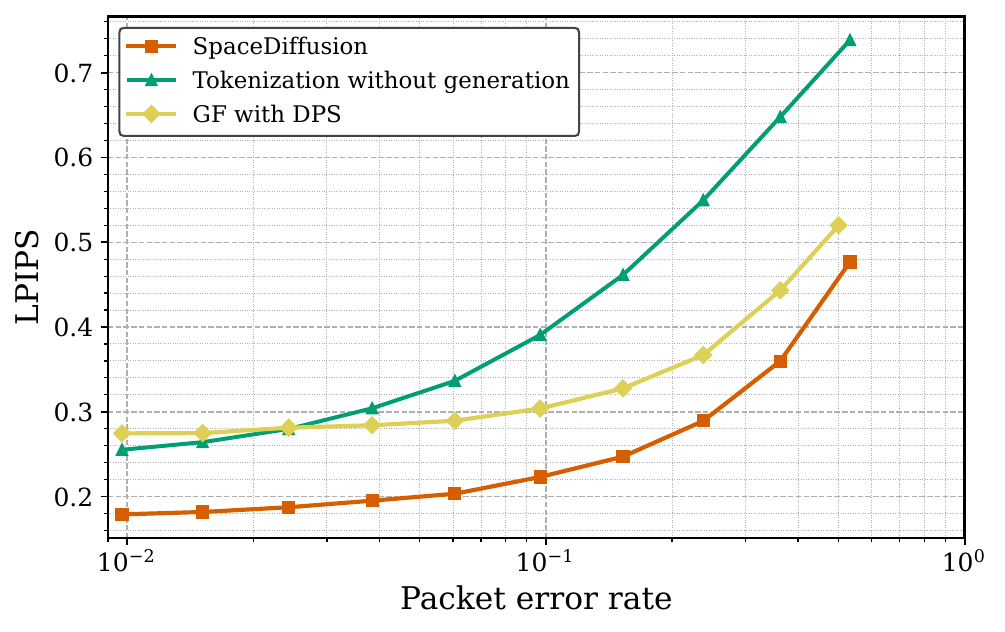}
	%\caption{fig2}
	\captionsetup{justification=justified}
	\caption{Reconstruction quality versus packet error rate.}
	\label{fig:per_lpips_comparison}
\end{figure} 

\begin{figure}[t]
	\normalsize
	\setlength{\abovecaptionskip}{4pt}
	\setlength{\belowcaptionskip}{-0.1cm}
	\centering
	\includegraphics[width=0.85\linewidth]{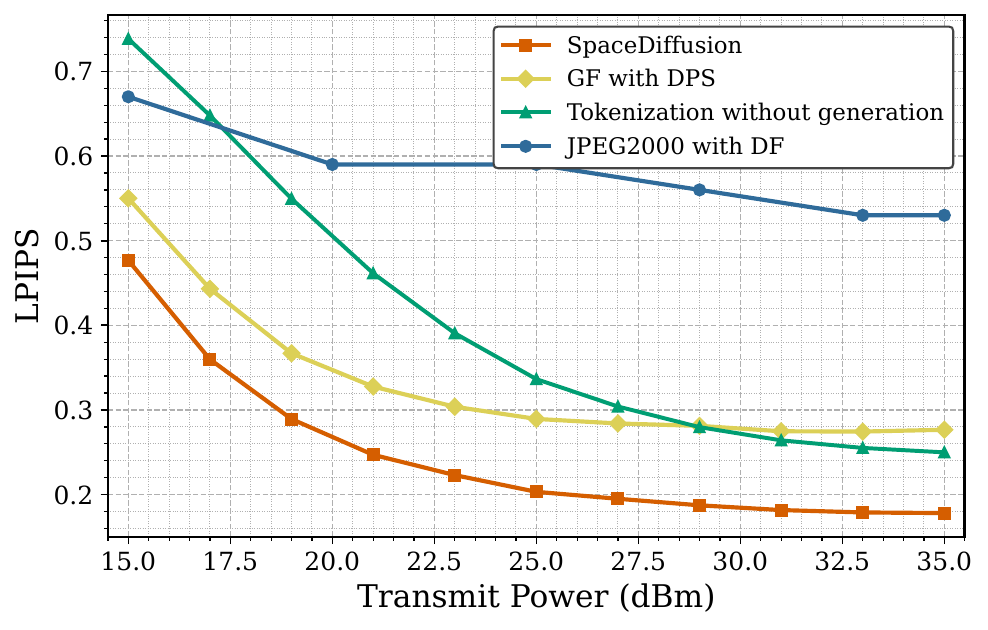}
	%\caption{fig2}
	\captionsetup{justification=justified}
	\caption{Reconstruction quality versus uplink transmit power. }
	\label{fig:power_lpips_comparison}
	 \vspace{-3pt}
\end{figure}                                                                               

\begin{figure}[t]
	\normalsize
	\setlength{\abovecaptionskip}{4pt}
	\setlength{\belowcaptionskip}{-0.1cm}
	\centering
	\includegraphics[width=1.\linewidth]{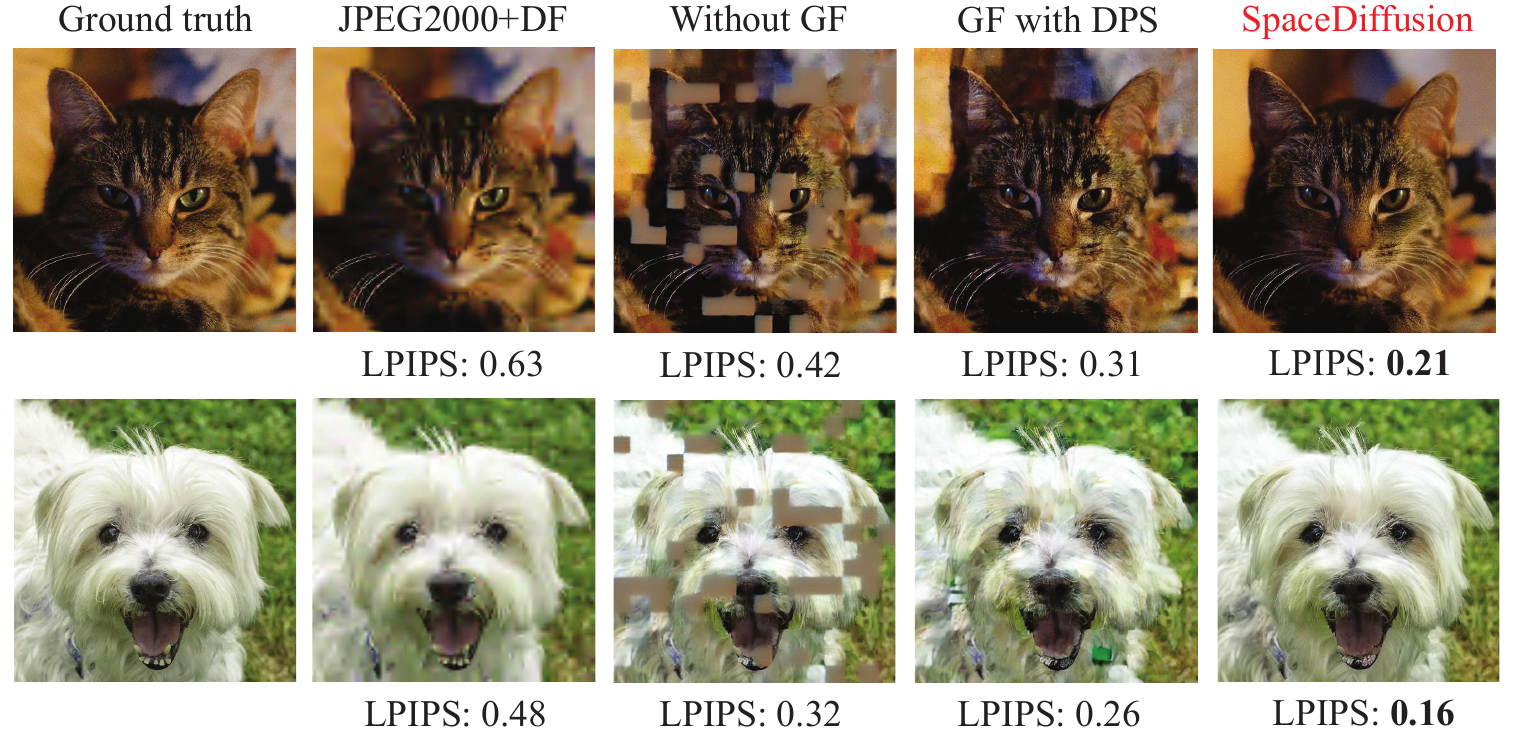}
	%\caption{fig2}
	\captionsetup{justification=justified}
	\caption{Visual reconstruction examples on the AFHQ dataset. }
	\label{fig:afhq_visual_examples}
	 \vspace{-3pt}
\end{figure}    

\begin{figure*}[t]
	\normalsize
	\setlength{\abovecaptionskip}{4pt}
	\setlength{\belowcaptionskip}{-0.1cm}
	\centering
	\includegraphics[width=1.\linewidth]{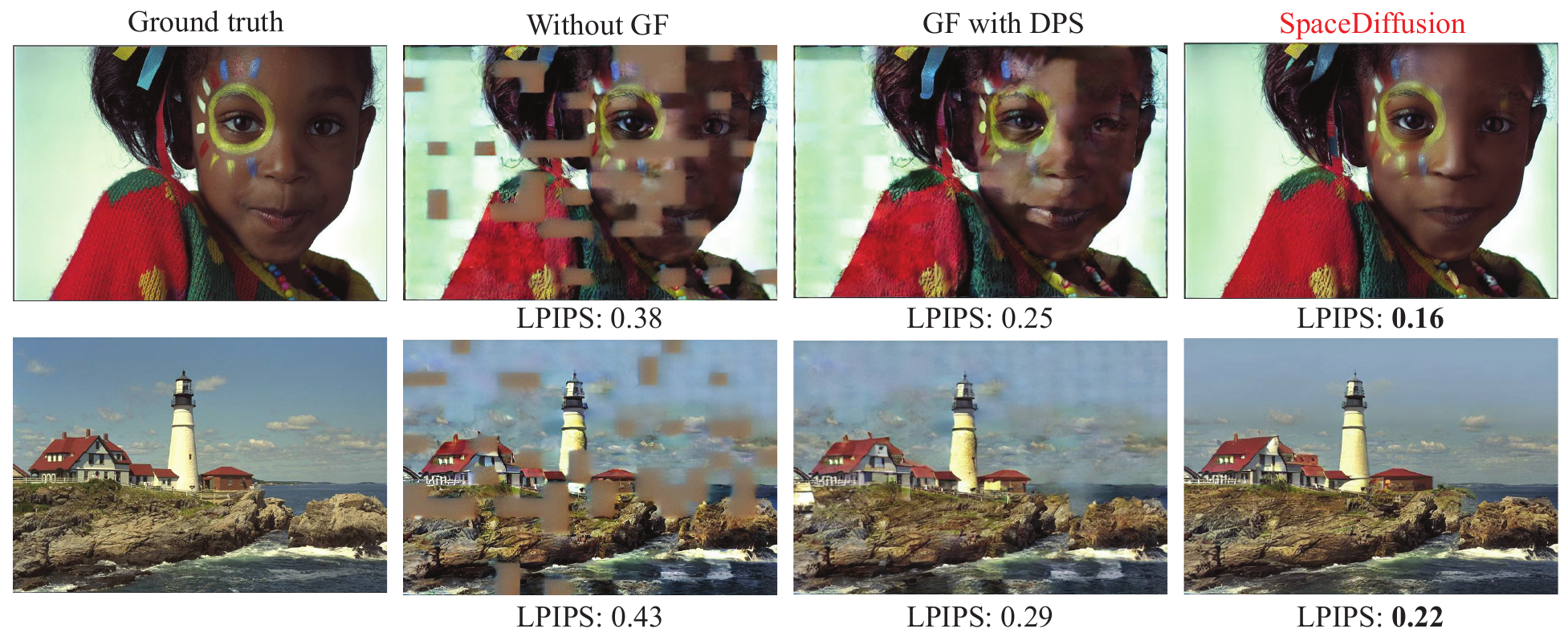}
	%\caption{fig2}
	\captionsetup{justification=justified}
	\caption{Visual reconstruction examples on the Kodak dataset \cite{Kodak}. }
	\label{fig:kodak_visual_examples}
	 \vspace{-7pt}
\end{figure*}    

Fig.~\ref{fig:per_lpips_comparison} shows the LPIPS performance as a function of the packet error rate. The retransmission protocol is disabled to evaluate the robust performance of all the schemes. Here, we omit the  JPEG2000 scheme  because its robust performance is much poorer than the other schemes. It is observed that all the benchmarking schemes degrade under  severe packet losses, but SpaceDiffusion consistently achieves the lowest LPIPS. For example, at a packet error rate of $0.1$, SpaceDiffusion obtains an LPIPS of about $0.22$, while GF with DPS and tokenization without generation achieve about $0.30$ and $0.39$, respectively. The performance gap becomes more pronounced in the medium-to-high error regime, where the non-generative tokenization baseline suffers from unrecovered token erasures. This demonstrates that the proposed SpaceDiffusion scheme can effectively recover the missing content and improve the E2E image quality under lossy uplink channels. 

Fig.~\ref{fig:power_lpips_comparison} further evaluates the performance as a function of the uplink transmit power. As the transmit power increases, the channel becomes more reliable and the LPIPS of all token-based schemes decreases. The proposed SpaceDiffusion scheme requires much lower transmit power to achieve high-quality image transmission. For example, when the target LPIPS is around $0.25$, SpaceDiffusion saves about $15$ dB of uplink transmit power compared with the non-generative tokenization baseline. This power saving makes the proposed scheme particularly suitable for power-constrained NTN IoT terminals and handheld devices.  

To examine generative fidelity and potential hallucinations, Figs.~\ref{fig:afhq_visual_examples} and \ref{fig:kodak_visual_examples} show reconstruction examples on the AFHQ and Kodak datasets at an uplink rate of approximately $0.11$ bpp and a packet error rate of $0.1$. JPEG2000 with retransmission produces visible blurring, while tokenization without generation leaves structured gaps caused by packet erasures.  DPS partially fills the missing regions but exhibits texture inconsistency and residual artifacts. In comparison, SpaceDiffusion produces more coherent structures and textures and achieves the lowest LPIPS on both datasets. Although the generated local details may not match the source exactly, SpaceDiffusion helps preserve enough semantic information for human interpretation.

\subsection{Discussions on Energy-Limited Scenarios}

\begin{figure*}[!t]
	\normalsize
	\setlength{\abovecaptionskip}{4pt}
	\setlength{\belowcaptionskip}{-0.1cm}
	\centering
	\subfigure[Solar power and optimized diffusion steps over orbital time.]{
		\includegraphics[width=0.45\textwidth]{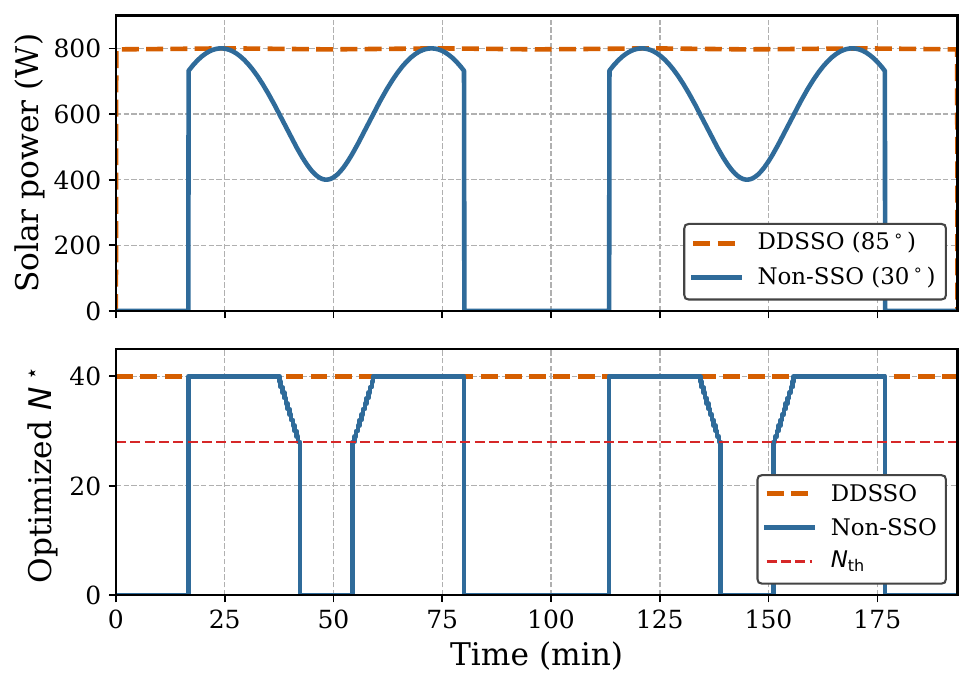}
		\label{fig:orbit_power_n_star_comparison}
	}
	\hfill
	\subfigure[Average LPIPS versus solar beta angle.]{
		\includegraphics[width=0.45\textwidth]{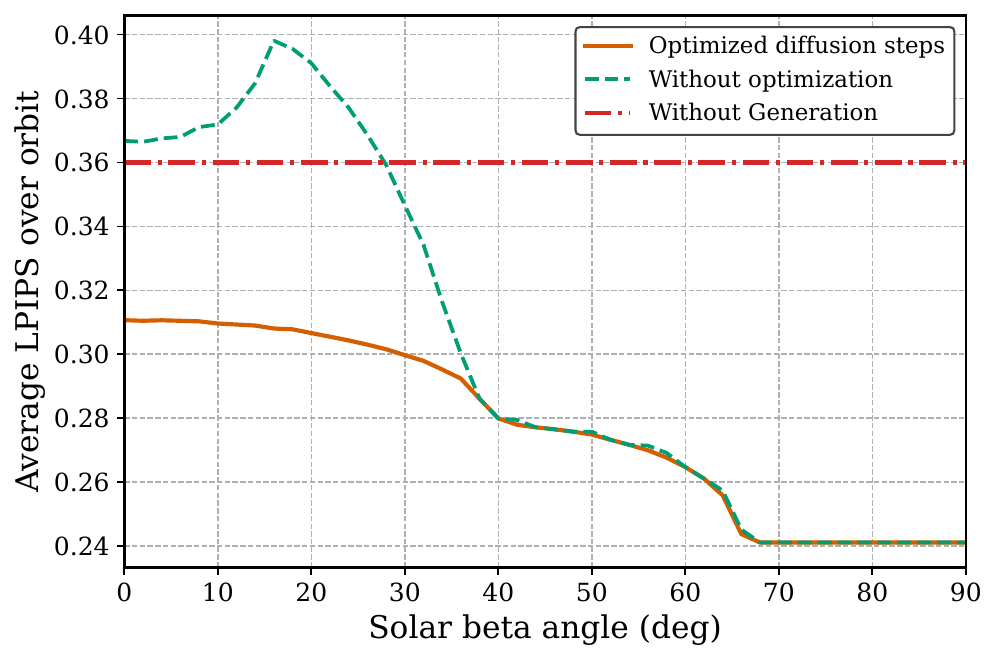}
		\label{fig:average_lpips_vs_alpha}
	}
	\captionsetup{justification=justified}
	\caption{Energy-limited diffusion scheduling.  }
	\label{fig:energy_limited_scenario}
	 \vspace{-7pt}
\end{figure*}

We finally examine whether SpaceDiffusion remains practical when the satellite energy supply is constrained by orbital illumination. Fig.~\ref{fig:orbit_power_n_star_comparison} compares a dawn-dusk Sun-synchronous orbit (DDSSO) and a non-Sun-synchronous orbit (Non-SSO). The DDSSO case has nearly continuous solar illumination, and the harvested power stays close to the peak value. Here, we set the maximum diffusion step size to be $N=40$. The satellite height is around $600$ km.  The maximum executable time and initial battery capacity are set as $8$ s and $0$ mAh, respectively.

% It is observed that the optimized number of reverse steps remains at $N^{\star}=40$, which is above the threshold $N_{\rm th}$ required to outperform the DF baseline. In contrast, the Non-SSO case experiences eclipse intervals. When the available energy is sufficient, the satellite executes the full reverse process; when the energy budget cannot support at least $N_{\rm th}$ steps, the optimizer sets $N^{\star}=0$ and the system degenerates to DF scheme. This threshold behavior avoids wasting energy on an incomplete diffusion process whose output may still be worse than direct forwarding.
% Fig.~\ref{fig:average_lpips_vs_alpha} shows the average LPIPS over an orbit duration as the solar beta angle increases. A larger beta angle reduces the eclipse duration and increases the effective energy available for on-board generation. The optimized diffusion-step policy consistently outperforms the no-generation baseline, reducing the average LPIPS from about $0.36$ to about $0.32$ in energy-scarce cases and to about $0.24$ when the beta angle is sufficiently large. Without the proposed optimization, the generation process can be activated even when the available energy supports only an insufficient number of reverse steps. This leads to poor intermediate reconstructions and can be worse than no generation in low-energy regimes. This figure also provides an interesting guideline for SpaceDiffusion deployment over different orbits. For example, if an average LPIPS target of $0.26$ is required, the satellite orbit should have a solar beta angle larger than approximately $64^\circ$.

Fig.~\ref{fig:orbit_power_n_star_comparison} shows that the nearly continuous illumination in DDSSO supports the full $N^{\star}=40$ reverse steps, whereas the Non-SSO case executes diffusion only when the available energy supports at least $N_{\rm th}$ steps; otherwise, it sets $N^{\star}=0$ and falls back to DF. This threshold policy prevents ineffective partial generation. As shown in Fig.~\ref{fig:average_lpips_vs_alpha}, a larger solar beta angle shortens the eclipse and increases the energy available for generation. The optimized policy consequently reduces the average LPIPS from about $0.36$ to $0.32$ in energy-scarce cases and to $0.24$ at sufficiently large beta angles, while avoiding the poor intermediate reconstructions produced by ungated generation.

\section{Concluding Remarks}
This paper proposed the SpaceDiffusion framework that exploits on-board computing resources for efficient image transmission over limited satellite links. An interesting insight is that the channel-distortion model can act as an inserted ``prompt'' that guides diffusion reconstruction. Thus, the GF framework can reuse well-trained latent diffusion models developed for image-generation tasks without channel-specific retraining. Simulation results further show that on-board generation significantly reduces E2E latency because it replaces feedback-driven uplink retransmissions with generative reconstruction. This latency reduction particularly benefits delay-sensitive IoT and handheld applications.

Several directions merit further study. SpaceDiffusion can be extended to multi-user and multi-satellite networks through joint communication, computation, and energy allocation. To mitigate generative hallucinations, future work may investigate unequal error protection for semantically important data tokens. Other promising topics include multimodal GF for video, audio, and sensing data and hardware-aware acceleration for practical on-board deployment.
\section*{Appendix}
\setcounter{section}{0}
\renewcommand{\thesection}{\Alph{section}}
\renewcommand{\thesectiondis}{\Alph{section}}
\subsection{Proof of Proposition \ref{prop:posterior_score_spacediffusion}}

	Let $a_n\in\{0,1\}$ denote the common entry of the $n$-th row of $\mathbf{A}$, i.e.,
	\begin{equation}
		[\mathbf{A}]_{n,d}=a_n,\qquad \forall d\in\{1,\dots,D\}.
	\end{equation}
	Also, let $\tilde{\mathbf{y}}_n$ and $\check{\mathbf{y}}_{t,n}$ denote the $n$-th rows of $\tilde{\mathbf{Y}}$ and $\check{\mathbf{Y}}_t$, respectively.

	Under the pseudo-inverse Gaussian approximation in \eqref{gaussian}, 
	each row satisfies
	\begin{equation}
		\mathbf{y}_n\mid \mathbf{Y}_t
		\approx
		\mathcal{N}(\check{\mathbf{y}}_{t,n},r_t^2\mathbf{I}_D).
	\end{equation}

	According to the inverse model in \eqref{inverse}, the row-wise observation model is
	\begin{equation}
		\tilde{\mathbf{y}}_n=
		\begin{cases}
			\mathbf{y}_n+\mathbf{q}_n, & a_n=1,\\
			\mathbf{0}, & a_n=0,
		\end{cases}
	\end{equation}
	where $\mathbf{q}_n\sim\mathcal{N}(\mathbf{0},\mathbf{\Sigma}_q)$.
	Therefore, when $a_n=1$,
	\begin{align}
		p_t(\tilde{\mathbf{y}}_n\mid \mathbf{Y}_t)
		&=
		\int p_t(\mathbf{y}_n\mid \mathbf{Y}_t)\,p(\tilde{\mathbf{y}}_n\mid \mathbf{y}_n)\,d\mathbf{y}_n \nonumber\\
		&\approx
		\int
		\mathcal{N}(\mathbf{y}_n;\check{\mathbf{y}}_{t,n},r_t^2\mathbf{I}_D)
		\mathcal{N}(\tilde{\mathbf{y}}_n;\mathbf{y}_n,\mathbf{\Sigma}_q)
		\,d\mathbf{y}_n \nonumber\\
		&=
		\mathcal{N}(\tilde{\mathbf{y}}_n;\check{\mathbf{y}}_{t,n},r_t^2\mathbf{I}_D+\mathbf{\Sigma}_q).
	\end{align}
	When $a_n=0$, $\tilde{\mathbf{y}}_n=\mathbf{0}$ is deterministic and thus contributes only a constant independent of $\mathbf{Y}_t$.

	Hence, up to an additive constant $c$ independent of $\mathbf{Y}_t$,
	\begin{align}
		&\log p_t(\tilde{\mathbf{Y}}\mid \mathbf{Y}_t) \nonumber \\
		&=
		c
		-\frac{1}{2}\sum_{n=1}^{K}
		a_n
		(\tilde{\mathbf{y}}_n-\check{\mathbf{y}}_{t,n})
		(r_t^2\mathbf{I}_D+\mathbf{\Sigma}_q)^{-1}
		(\tilde{\mathbf{y}}_n-\check{\mathbf{y}}_{t,n})^{\top}.
	\end{align}

	Differentiating the above expression with respect to $\check{\mathbf{Y}}_t$ yields
	\begin{align}
		\nabla_{\check{\mathbf{Y}}_t}\log p_t(\tilde{\mathbf{Y}}\mid \mathbf{Y}_t)
		=
		\mathbf{A}\odot
		\left(
			(\tilde{\mathbf{Y}}-\check{\mathbf{Y}}_t)
			(r_t^2\mathbf{I}_D+\mathbf{\Sigma}_q)^{-1}
		\right).
	\end{align}

	Then, by the chain rule in vectorized form, we have
	\begin{align}
		&\operatorname{vec}\!\left(
			\nabla_{\mathbf{Y}_t}\log p_t(\tilde{\mathbf{Y}}\mid \mathbf{Y}_t)
		\right) \nonumber \\
		&=
		\left(
			\frac{\partial \operatorname{vec}(\check{\mathbf{Y}}_t)}
			{\partial \operatorname{vec}(\mathbf{Y}_t)}
		\right)^{\top}
		\operatorname{vec}\!\left(
			\nabla_{\check{\mathbf{Y}}_t}\log p_t(\tilde{\mathbf{Y}}\mid \mathbf{Y}_t)
		\right) \nonumber\\
		&=
		\mathbf{J}_t^{\top}\operatorname{vec}(\mathbf{G}_t).
	\end{align}
	Equivalently, we have
	$
		\nabla_{\mathbf{Y}_t}\log p_t(\tilde{\mathbf{Y}}\mid \mathbf{Y}_t)
		=
		\operatorname{unvec}\!\left(
			\mathbf{J}_t^{\top}\operatorname{vec}(\mathbf{G}_t)
		\right).
	$
Finally, by substituting the above result into \eqref{score}, 
	we obtain
	\begin{align}
		\nabla_{\mathbf{Y}_t}\log p_t(\mathbf{Y}_t\mid \tilde{\mathbf{Y}})
		\approx
		\mathbf{S}_{\bm{\Theta}}(\mathbf{Y}_t,t)
		+
		\operatorname{unvec}\!\left(
			\mathbf{J}_t^{\top}\operatorname{vec}(\mathbf{G}_t)
		\right).
	\end{align}
	This completes the proof.

\subsection{Proof of Proposition \ref{thm:reconstruction_error_upper_bound}}\label{appendix:proof_thm_reconstruction_error_upper_bound}

	Under the deterministic DDIM setting $\sigma_t=0$, we have 
	\begin{align}
		\mathbf{Y}_{t-1}
		=
		\sqrt{\bar{\alpha}_{t-1}}\check{\mathbf{Y}}_t
		+
		\sqrt{1-\bar{\alpha}_{t-1}}\hat{\boldsymbol{\epsilon}}_{\bm{\Theta},t}
		+
		\delta_t\mathcal{L}_t.
	\end{align}
	Let
$
		\mathbf{e}_t
		=
		\mathbf{Y}_{t-1}^{\mathrm{ddim}}-\mathbf{Y}.
$
	Then, the reconstruction error can be calculated as 
	\begin{align}
		\mathcal{E}_{t-1}
		&=
		\mathbb{E}\!\left[\|\mathbf{Y}_{t-1}-\mathbf{Y}\|_F^2\right] \nonumber\\
		&=
		\mathbb{E}\!\left[
		\|\mathbf{e}_t+\delta_t\mathcal{L}_t\|_F^2
		\right] \nonumber\\
		&=
		\mathbb{E}\!\left[\|\mathbf{e}_t\|_F^2\right]
		+
		2\delta_t
		\mathbb{E}\!\left[
		\left\langle\mathbf{e}_t,\mathcal{L}_t\right\rangle
		\right]
		+
		\delta_t^2
		\mathbb{E}\!\left[\|\mathcal{L}_t\|_F^2\right] \nonumber\\
		&\overset{(a)}{\le}
		\left(
		\sqrt{\mathbb{E}\!\left[\|\mathbf{e}_t\|_F^2\right]}
		+
		\delta_t
		\sqrt{\mathbb{E}\!\left[\|\mathcal{L}_t\|_F^2\right]}
		\right)^2 \nonumber\\
		&\overset{(b)}{\approx}
		\left(
		\sqrt{a_t^{\mathrm{ddim}}}
		+
		\delta_t
		\sqrt{\mathbb{E}\!\left[\|\mathcal{L}_t\|_F^2\right]}
		\right)^2 \nonumber\\
		&\overset{(c)}{\lesssim}
		\left(
		\sqrt{a_t^{\mathrm{ddim}}}
		+
		\delta_t\kappa\sqrt{\chi_t}
		\right)^2,
		\label{eq:appendix_main_error_chain}
	\end{align}
	where $\langle\cdot,\cdot\rangle$ denotes the Frobenius inner product. 
		The justifications of the steps are as follows.
		\begin{itemize}
			\item \textbf{Step $(a)$:} By the Cauchy--Schwarz inequality \cite{rohatgi2015introduction}, we have
			\begin{equation}
				\left|
				\mathbb{E}\!\left[
				\left\langle\mathbf{e}_t,\mathcal{L}_t\right\rangle
				\right]
				\right|
				\le
				\sqrt{
				\mathbb{E}[\|\mathbf{e}_t\|_F^2]
				\mathbb{E}[\|\mathcal{L}_t\|_F^2]
				}.
			\end{equation}
			
					\item \textbf{Step $(b)$:} From \eqref{eq:checkY_from_noise} and Assumption \ref{assump:denoising_error}, the conventional DDIM error can be written as
					\begin{align}
						\mathbf{e}_t
						&=
						\frac{\sqrt{\bar{\alpha}_{t-1}}}{\sqrt{\bar{\alpha}_t}}\mathbf{Y}_t
						-
						\mathbf{Y} \nonumber\\
						&\quad+
						\left(
						\sqrt{1-\bar{\alpha}_{t-1}}
						-
						\sqrt{\frac{\bar{\alpha}_{t-1}(1-\bar{\alpha}_t)}{\bar{\alpha}_t}}
						\right)
						(\boldsymbol{\epsilon}_t+\mathbf{U}_t).
						\label{eq:appendix_et_with_yt}
					\end{align}
					Using the forward-trajectory approximation in \eqref{eq:ddim_forward},
					\begin{equation}
						\mathbf{Y}_t
						\approx
						\sqrt{\bar{\alpha}_t}\mathbf{Y}
						+
						\sqrt{1-\bar{\alpha}_t}\boldsymbol{\epsilon}_t.
					\end{equation}
					substitution into \eqref{eq:appendix_et_with_yt} gives
					\begin{align}
						\mathbf{e}_t
						&\approx
						(\sqrt{\bar{\alpha}_{t-1}}-1)\mathbf{Y}
						+
						\sqrt{1-\bar{\alpha}_{t-1}}\boldsymbol{\epsilon}_t \nonumber\\
						&\quad+
						\left(
						\sqrt{1-\bar{\alpha}_{t-1}}
						-
						\sqrt{\frac{\bar{\alpha}_{t-1}(1-\bar{\alpha}_t)}{\bar{\alpha}_t}}
						\right)\mathbf{U}_t.
					\end{align}
					Since $\boldsymbol{\epsilon}_t$ and $\mathbf{U}_t$ are zero-mean and independent of $\mathbf{Y}$, taking expectations gives $\mathbb{E}[\|\mathbf{e}_t\|_F^2]\approx a_t^{\mathrm{ddim}}$ as defined in \eqref{eq:ddim_error_with_training_noise}.
			
			\item \textbf{Step $(c)$:} For the $n$-th row $\mathbf{g}_{t,n}$ of $\mathbf{G}_t$, Proposition \ref{prop:posterior_score_spacediffusion} gives
				\begin{align}
					\mathbf{g}_{t,n}
					=
						a_n(\tilde{\mathbf{y}}_n-\check{\mathbf{y}}_{t,n})\mathbf{B}_t 
						=
						a_n\left(\mathbf{q}_n+\sqrt{\frac{1-\bar{\alpha}_t}{\bar{\alpha}_t}}\mathbf{u}_{t,n}\right)\mathbf{B}_t,
					\end{align}
					where $\tilde{\mathbf{y}}_n=\mathbf{y}_n+\mathbf{q}_n$ for $a_n=1$ and $\check{\mathbf{y}}_{t,n}\approx\mathbf{y}_n-\sqrt{(1-\bar{\alpha}_t)/\bar{\alpha}_t}\,\mathbf{u}_{t,n}$. Since $a_n$, $\mathbf{q}_n$, and $\mathbf{u}_{t,n}$ are independent,
			$
				\mathbb{E}\!\left[\|\mathbf{G}_t\|_F^2\right]
				=
				\chi_t.
			$
			Moreover, from
			$
				\mathcal{L}_t
				=
				\operatorname{unvec}
				(\mathbf{J}_t^{\top}\operatorname{vec}(\mathbf{G}_t))
			$
			and $\|\mathbf{J}_t\|_2\le\kappa_t$,
			\begin{align}
				\|\mathcal{L}_t\|_F^2
				&=
				\left\|
				\mathbf{J}_t^{\top}\operatorname{vec}(\mathbf{G}_t)
				\right\|_2^2 \nonumber\\
				&\le
				\|\mathbf{J}_t\|_2^2\|\mathbf{G}_t\|_F^2
				\le
				\kappa^2\|\mathbf{G}_t\|_F^2.
			\end{align}
			Taking expectation yields $\mathbb{E}[\|\mathcal{L}_t\|_F^2]\le\kappa^2\chi_t$. 
		\end{itemize}
		This completes the proof.

\addtolength{\topmargin}{0.02in}	
\IEEEpeerreviewmaketitle
\bibliographystyle{IEEEtran}
\bibliography{semantic}

\end{document}